\documentclass[a4paper,reqno]{amsart}

\usepackage[utf8]{inputenc}
\usepackage[english]{babel} 
\usepackage{a4wide}

\usepackage{amsmath,amssymb,amsfonts,amsthm}
\usepackage{mathrsfs, esint, dsfont}
\usepackage{moreverb,rotating,graphics,float}
\usepackage[colorlinks, linkcolor={blue!50!blue}, citecolor={red}]{hyperref}
\usepackage{framed,fancybox}
\usepackage{enumerate}
\usepackage{csquotes}
\usepackage{todonotes}

\usepackage{subcaption}

\newtheorem{theorem}{Theorem}[section]
\newtheorem{proposition}[theorem]{Proposition}
\newtheorem{remark}[theorem]{Remark}
\newtheorem{lemma}[theorem]{Lemma}

\newtheorem{definition}[theorem]{Definition}

\newcommand\1{{\mathds 1}}

\def\C{{\mathbb C}}

\def\bbI{{\mathbb I}}

\def\N{{\mathbb N}}

\def\R{{\mathbb R}}

\def\SS{{\mathbb S}}
\def\TT{{\mathbb T}}

\def\Z{{\mathbb Z}}

\def\bc{{\mathbf c}}

\def\bs{{\mathbf s}}

\def\bu{{\mathbf u}}

\def\bone{{\mathbf 1}}

\def\bsigma{{\boldsymbol \sigma}}

\def\rd{{\mathrm{d}}}
\def\re{{\mathrm{e}}}
\def\ri{{\mathrm{i}}}

\def\cB{{\mathcal B}}

\def\cD{{\mathcal D}}

\def\cI{{\mathcal I}}

\def\cL{{\mathcal L}}
\def\cM{{\mathcal M}}

\def\cS{{\mathcal S}}

\def\cZ{{\mathcal Z}}

\newcommand{\sC}{\mathscr{C}}

\newcommand{\loc}{{\rm loc}}
\newcommand{\Tr}{\rm Tr \,}
\newcommand{\ess}{{\rm ess}}
\newcommand{\per}{{\rm per}}
\newcommand{\Ch}{{\rm Ch}}
\newcommand{\Span}{{\rm Span}}
\renewcommand{\Re}{{\rm Re}}

\newcommand{\bra}{\langle}
\newcommand{\ket}{\rangle}
\newcommand{\U}{{\rm U}}
\newcommand{\Ran}{{\rm Ran}}

\newcommand{\up}{\uparrow}
\newcommand{\down}{\downarrow}

\author{David Gontier}
\address{CEREMADE, University of Paris-Dauphine, PSL University, 75016 Paris, France}

\email{gontier@ceremade.dauphine.fr}

\title[Edge states in ODEs for dislocations]{Edge states in Ordinary Differential Equations for dislocations}

\date{\today}

\begin{document}
    
    \begin{abstract}
        In this article, we study Schrödinger operators on the real line, when the external potential represents a dislocation in a periodic medium. We study how the spectrum varies with the dislocation parameter. We introduce several integer-valued indices, including Chern number for bulk indices, and various spectral flows for edge indices. We prove that all these indices coincide, providing a proof a bulk-edge correspondence in this case. The study is also made for dislocations in Dirac models on the real line. We prove that 0 is always an eigenvalue of such operators.
    \end{abstract}

\maketitle


\section{Introduction}

In material science, bulk-edge correspondence enables to link integer-valued indices computed in the bulk (material on a full space), with indices computed on an edge (material on a half-space). The first proof that the two indices coincide appears in the work of Hatsugai~\cite{hatsugai1993chern}. There is now a variety of proofs on different contexts, using complex theory and/or $K$-theory (see {\em e.g.}~\cite{prodan2016bulk}). In the present article, we provide an elementary proof of bulk-edge correspondence in a simple continuous one dimensional setting in the context of dislocations. The study is done for Schrödinger and Dirac models.

\subsection*{Main results in the Schrödinger case}
In the first part of this article, we focus on three families of self-adjoint Schrödinger Hamiltonian: 
\begin{itemize}
    \item the bulk Hamiltonian
    \[
    H(t) := - \partial_{xx}^2 + V(x - t), \quad \text{acting on $L^2(\R)$, with domain $H^2(\R)$};
    \]
    \item the domain wall edge Hamiltonian
    \[
        H^\sharp_\chi(t) = H(0) \chi + H(t) (1 - \chi)  \quad
        \text{acting on $L^2(\R)$, with domain $H^2(\R)$};
    \]
    \item the Dirichlet edge Hamiltonian
    \[
        H^\sharp_D(t) := - \partial_{xx}^2 + V(x - t), \quad \text{acting on $L^2(\R^+)$, with domain $H^2_0(\R^+)$}.
    \]
\end{itemize}
Here, $V(x)$ is a real-valued $1$-periodic potential, and $\chi(x)$ is a bounded switch function satisfying $\chi(x) = 1$ for $x < -L$ and $\chi(x) = 0$ for $x > L$, with $L> 0$ large enough. 

\medskip

The operator $H^\sharp_\chi(t)$ describes a dislocation between a fixed left Hamiltonian $H(0)$ and a translated version of it on the right $H(t)$. This model was studied in~\cite{korotyaev2000lattice, korotyaev2005schrodinger} in the case $\chi(x) = \1(x < 0)$ using complex analysis, and more recently in~\cite{fefferman2017topologically, drouot2018defect, drouot2018bulk} in a perturbative regime. Here, we present a topological approach, which leads to similar conclusions to these articles, in a slightly more general setting.

The main goal of this article is to introduce several indices, and prove that they are all equal. We sum up here the different indices that we introduce. Since $V$ is $1$-periodic, the operators $H(t)$, $H^\sharp_\chi(t)$ and $H^\sharp_D(t)$ are all $1$-periodic in $t$. In addition, it is a classical result that their essential spectra coincide. Actually,
\[
    \sigma \left( H (t)\right) = \sigma_{\ess} \left( H (t)\right)   =  \sigma_{\ess} \left( H^\sharp_\chi (t)\right) =  \sigma_{\ess} \left( H^\sharp_D (t)\right) = \sigma \left( H (0)\right) = \bigcup_{n \ge 1} [E_n^-, E_n^+],
\]
where $E_1^- < E_1^+ \le E_2^- < E_2^+ \le E_3^- < E_3^+ \le \cdots$ are the band edges. For $n \ge 1$, the $n$-th essential gap is the open interval $g_n := (E_n^+, E_{n+1}^-)$, and the $0$-th gap is $g_0 := (-\infty, E_1^-)$. We say that the $n$-th gap is open if $E_n^+ < E_{n+1}^-$, and is empty otherwise. We assume that the $n$-th gap is open, and we consider $E \in g_n$ in this gap.

\medskip

\noindent \underline{\bf Bulk index}. In Section~\ref{ssec:bulk_index}, we treat the (bulk) equations $-u'' + V(x - t) u = E u$ as a $1$-periodic family of ordinary differential equations (ODEs) depending on $t$. For each $t \in \TT^1$, the vectorial space of solutions $\cL(t)$ is of dimension $2$. There is a natural splitting between the solutions: the ones that decay exponentially at $+ \infty$ (in $\cL^+(t)$), and the one that decays exponentially at $-\infty$ (in $\cL^-(t)$). The map $t \mapsto \cL^+(t)$ is $1$-periodic, and we associate to it a Maslov index $\cB_n \in \Z$.

\medskip

\noindent \underline{\bf Chern number}. In Section~\ref{ssec:Chern}, we focus on the operator $H(t)$, and we consider the projector on the $n$ lowest bands $P_n(t) := \1 (H(t) \le E)$. Since $P_n(t)$ commutes with translations, we can Bloch transform it, and obtain a family of rank-$n$ projectors $P_n(t,k)$ acting on $L^2([0,1])$, which is periodic in both $t$ and $k$. For such periodic family of projectors we associate a Chern number $\Ch (P_n) \in \Z$.

\medskip

\noindent \underline{\bf Edge index}. In Section~\ref{ssec:edgeIndex}, we see the (edge) equation $H^\sharp_\chi(t) u = Eu$ as a $1$-periodic family of ODEs. We introduce the vectorial spaces $\cL^{\sharp, \pm}_\chi(t)$ of solutions that decay at $\pm \infty$. These spaces may cross, and if $u \in \cL^{\sharp, +}_\chi(t) \cap \cL^{\sharp, -}_\chi(t)$, then $u$ is an eigenvector for $H_\chi^\sharp(t)$, that is an edge state. We associate to such bi-family of vectorial spaces an edge index $\cI^\sharp_{\chi,n} \in \Z$.

\medskip

\noindent \underline{\bf Domain wall spectral flow}. In Section~\ref{ssec:spectralFlowDW}, we focus on the edge operator $H_\chi^\sharp(t)$. Although its essential spectrum is independent of $t$, some eigenvalues may appear in the essential gaps. The spectral flow $\cS_{\chi, n}^\sharp \in \Z$ is the net flow of eigenvalues going {\em downwards} through the gap.

\medskip

\noindent \underline{\bf Dirichlet spectral flow}. Finally, in Section~\ref{ssec:spectralFlowDirichlet}, we consider the operator $H^\sharp_D(t)$. Again, its essential spectrum is independent of $t$, and some eigenvalues may appear in its essential gaps. We associate a spectral flow $\cS^\sharp_{D, n}$ to this family as well.

\begin{remark}
    As in~\cite{drouot2018bulk}, we chose the convention to count the flow of eigenvalues going {\em downwards} for the spectral flow. This allows a nicer statement of the following theorem.
\end{remark}

The main result of this article can be summarised as follows.

\begin{theorem}[Bulk-edge correspondence] \label{th:bulk-edge}
    If the $n$-th gap is open, then 
    \[
    \boxed{ \cB_n = \Ch(P_n) = \cI_{\chi,n}^\sharp = \cS_{\chi, n}^\sharp = \cS_{D, n}^\sharp = n. }
    \]
    In particular, the indices are independent of $\chi$ and of the choice of $E$ in the gap. In addition, all eigenvalues of $H^\sharp_\chi$ and $H^\sharp_D$ are simple, and the corresponding eigenstates are exponentially localised.
\end{theorem}

This result is more or less already known: the proof of Hatsugai~\cite{hatsugai1993chern} in the discrete case can be used in our continuous setting to prove $\Ch(P_n) = \cS^\sharp_{D, n}$. The equality $\cB_n = \cI^\sharp_{D, n}$ was proved in a discrete setting in~\cite{avila2013topological}. Combining the two results gives $\Ch(P_n) = \cB_n$, a fact noticed in~\cite{avila2013topological} (we also refer to the new articles~\cite{bal2017topological, bal2018continuous} for a study in a continuous two-dimensional setting). 
Recently, Drouot~\cite{drouot2018bulk} proved $\Ch(P_n) = \cS^{\sharp}_{\chi, n}$ in the continuous setting, under the extra condition that the $n$-th gap does not close under a particular deformation. Finally, the equality $\cS_{\chi, n}^\sharp = n$ was proved in~\cite{korotyaev2000lattice} for the special case $\chi = \1(x < 0)$ (see also the short proof in~\cite{hempel2011variational}). 

\medskip

In this article, we give new and elementary proofs that all these indices equal $n$.

\begin{remark} [Junction case]
    Most of the techniques introduced here can be used in more complex settings, for instance to study the junction between two periodic media. In this case, we study
    \[
    H_\chi^\sharp(t) = H_1(t) \chi + H_2(t) (1 -\chi), \quad \text{with} \quad H_i(t) := - \partial_{xx}^2 + V_{i,t}(x),
    \]
    where $V_{i,t}$ are periodic potentials in $x$ (with possibly different periods), and $t \mapsto V_{i,t}$ are $1$-periodic. We do not comment more on this fact, as it makes the study more tedious.
\end{remark}

\subsection*{Main results in the Dirac case}

In a second part of the article, we show how the same techniques can be applied to study dislocations for Dirac operators. We focus on two families of self-adjoint Dirac operators:
\begin{itemize}
    \item the bulk Dirac operator
    \[
    \cD(t) := ( - \ri \partial_x) \bsigma_3 + \re^{ - \ri \pi t \bsigma_3} \left[V(x) \bsigma_1 \right]  \re^{ \ri \pi t \bsigma_3} \quad \text{acting on $L^2(\R, \C^2)$, with domain $H^1(\R, \C^2)$};
    \]
    \item the domain wall Dirac operator
    \[
    \cD_\chi^\sharp(t) := \cD(0) \chi + \cD(t) (1 - \chi) \quad \text{acting on $L^2(\R, \C^2)$, with domain $H^1(\R, \C^2)$}.
    \]
\end{itemize}
Here, we introduced $\bsigma_1, \bsigma_2$ and $\bsigma_3$ the usual $2 \times 2$ Pauli matrices, and the functions $V$ and $\chi$ are chosen as before, although our results hold in much more general cases, see Remark~\ref{rem:generalPotentials} below. We do not consider a Dirichlet version in the Dirac case, as the corresponding operator is not self-adjoint. The operator $\cD^\sharp_\chi(t)$ describes a dislocation between a fixed Dirac operator $\cD(0)$ on the left and a spin-rotated version of it $\cD(t)$ on the right. 

The operators $\cD(t)$ and $\cD^\sharp_\chi(t)$ are $1$-periodic in $t$. In addition, their essential spectra coincide. Actually, we have
\[
    \sigma ( \cD(t)) = \sigma_{\ess} \left( \cD(t) \right) = \sigma_{\ess} \left( \cD^\sharp_\chi(t) \right) = \sigma(\cD(0)).
\]
We consider $g \subset \R$ an open gap of $\cD(0)$ (we assume that there is at least one), and $E \in g$. 

\medskip

\noindent \underline{\bf Bulk index}. In Section~\ref{ssec:Dirac_bulkIndex}, we see the bulk equations $\cD(t) u = Eu$ as a $1$-periodic family of ODEs. The vectorial space of solutions $\cL(t)$ is again of dimension $2$ (over the field $\C$). We can consider $\cL^\pm(t)$ the sub-vectorial spaces of solutions that are exponentially decaying at $\pm \infty$. We associate a Maslov index $\cB \in \Z$ to the family $t \mapsto {\cL^+}(t)$.

\medskip

\noindent \underline{\bf Edge index}. In Section~\ref{ssec:Dirac_edgeIndex}, we see the edge equations $\cD^\sharp(t)u = Eu$ as a $1$-periodic family of ODEs. We have again a splitting ${\cL^{\sharp,+}_\chi}(t)$ and ${\cL^{\sharp,-}_\chi}(t)$, and these spaces may cross. If this happens, any $u$ in the intersection is an {\em edge state}. We associate to this bi-family an edge index ${\cI^\sharp_{\chi}} \in \Z$.

\medskip

\noindent \underline{\bf Spectral flow}. In Section~\ref{ssec:Dirac_spectralFlow} we consider the spectral flow of eigenvalues of $\cD^\sharp_\chi(t)$ going downwards in the gap $g$. We denote it by ${\cS^\sharp_\chi}$.

\begin{theorem}[Bulk-edge correspondence for Dirac] For all open gaps $g \subset \R$ of $\cD(0)$, we have
    \[
        \boxed{ \cB = \cI^\sharp_\chi = \cS^\sharp_\chi = 1. }
    \]
    In particular, the indices are independent of $\chi$ and of the choice of $E$ in the gap.  In addition, all eigenvalues of $\cD^\sharp_\chi$ are simple, and the corresponding eigenstates are exponentially localised.
\end{theorem}

The most interesting case is the $t = \frac12$ one, since we have
\[
\cD(0) = ( - \ri \partial_x) \bsigma_3 + V(x) \bsigma_1 
\quad \text{and} \quad
\cD(\tfrac12) = ( - \ri \partial_x) \bsigma_3 - V(x) \bsigma_1.
\]
In this case, $\cD^\sharp_\chi(\tfrac12)$ describes a smooth transition between $V$ and $-V$. From this theorem together with the symmetry
$\sigma \left( \cD^\sharp_\chi (t)   \right) = - \sigma \left( \cD^\sharp_\chi (1 - t)   \right)$ (see Section~\ref{ssec:Dirac_t=12}),
we obtain the following result
\begin{theorem}
At $t = \tfrac12$, the spectrum of $\cD^\sharp_\chi(\tfrac12)$ is symmetric with respect to the origin. If in addition $0$ is not in the essential spectrum of $\cD(0)$, then $0$ is an eigenvalue of the domain wall Dirac operator $\cD^\sharp_\chi(\tfrac12)$. 
\end{theorem}
This was already proved in~\cite{fefferman2017topologically, drouot2018bulk, drouot2018defect}.
The main contribution of the present work is to embed the operator $\cD^\sharp_\chi(\tfrac12)$ in the continuous family of operators $\cD^\sharp_\chi(t)$. The {\em topologically protected state} mentioned is the previous works is seen here as a manifestation of a spectral flow.

\medskip

This article is organised as follows. In Section~\ref{sec:firstFacts}, we gather our notations and recall basic facts about ODEs and Hamiltonian operators on a line. We prove our results concerning the Schrödinger case in Section~\ref{sec:HamiltonianCase}, and the ones concerning the Dirac case in Section~\ref{sec:DiracCase}. Several numerical illustration are provided in these sections. For the sake of clarity, we postpone most of our proofs concerning regularity in Section~\ref{sec:proofs}. Some extra independent proofs are put in the Appendix for completeness.

\section{First facts and notation}
\label{sec:firstFacts}

\subsection{The winding number}
We start with a brief section about winding numbers, as we relate most of our indices to these objects. Although all results are well-known, we set here some notation and recall some proofs, for we use similar ideas in the sequel. We denote by $\TT^1$ the torus $[0, 1]$, and by $\SS^1 := \{ z \in \C, \, | z | = 1\}$ the unit complex circle. It is possible to identify $\TT^1$ with $\SS^1$, but we avoid doing so to emphasise that $t \in \TT^1$ is real-valued.

\medskip

If $u$ is a continuous function from $\TT^1$ to $\SS^1$, we can associate a {\em winding number} $W[u] \in \Z$. It counts the number of times $u(t)$ turns around $0 \in \C$ as $t$ goes from $0$ to $1$.
\begin{definition}
    Let $u$ be a continuous map from $\TT^1$ to $\SS^1$, and let $\alpha : [0, 1] \to \R$ be a continuous lifting of $u$, that is $u(t) = \re^{ \ri \alpha(t) }$ for all $t \in [0, 1]$. The winding number of $u$ is $W[u] := (2 \pi)^{-1} \left( \alpha (1) -  \alpha (0)  \right)$.
\end{definition}
By periodicity, we must have $\alpha(1) = \alpha(0) + 2 k \pi$ with $k \in \Z$, in which case $W[u]= k$. This proves that the winding number is indeed an integer. If $\alpha$ and $\widetilde{\alpha}$ are two continuous liftings, then $1 = u(t) u(t)^{-1} = \re^{ \ri (\alpha - \widetilde{\alpha})(t)}$ and the difference $\alpha - \widetilde{\alpha}$ is constant by continuity. This implies that the winding number is independent of the lifting. 

If $u(s,t) : [0, 1] \times \TT^1 \to \SS^1$ is continuous, and if $\alpha(s,t)$ is a continuous lifting of $u(s,t)$, then $W[u(s, \cdot)] := ( 2 \pi)^{-1} \left( \alpha(s, 1) - \alpha(s, 0) \right)$ is continuous and integer valued, hence is constant. In particular, $W[u(1, \cdot)]  = W[u(0, \cdot)]$. We deduce that the winding number of $u$ only depends on the homotopy class of $u$.

 If $u$ and $v$ are continuous from $\TT^1$ to $\SS^1$, then so are $u^{-1}$ and $uv$. We directly have from the definition $W[u^{-1}] = - W[u]$, and $W[uv] = W[u] + W[v]$. In other words, the maps $W$ is a group homomorphism from $\left[ C^0(\TT^1, \SS^1), \times \right]$ to $[ \Z, +]$.
 
For our purpose, we need the following characterisation, which is valid for continuously differentiable functions\footnote{If $u : \TT^1 \to \SS^1$ is only continuous, we can apply a convolution kernel to it to obtain a smooth function $\tilde{u}$. If the convolution is sharp enough, $\tilde{u}$ and $u$ are homotope, and the following definition can be applied to $\tilde{u}$ to obtain the winding of $u$.}. 
\begin{lemma}
    If $u : \TT^1 \to \SS^1$ is continuously differentiable, then
    \[
        W[u] = \dfrac{1}{2 \ri \pi} \int_0^1 \dfrac{u'(t)}{u(t)} \rd t.
    \]
\end{lemma}

\begin{proof}
    The tangent line of $\SS^1$ at $u(t)$ is $\ri u(t) \R$, so $ - \ri u'(t)/u(t)$ is real valued.
    Let $\alpha_0 \in \R$ be such that $u(0)= \re^{ \ri \alpha_0}$, and define $\alpha(t) := \alpha_0 - \ri  \int_0^t \frac{u'}{u}(s) \rd s$.
    Then, $\alpha$ is a well-defined function which is continuously differentiable on $\R$. Moreover, it holds $u = \re^{\ri \alpha}$, so $\alpha$ is a continuously differentiable lifting of $u$. For this lifting, we have
    \[
        W[u] = \dfrac{1}{2 \pi} \left( \alpha(1) - \alpha(0) \right)= \dfrac{1}{2 \pi} \int_0^1 \alpha'(t) \rd t = \dfrac{1}{2 \ri \pi} \int_0^1 \frac{u'(t)}{u(t)} \rd t.
    \]
\end{proof}
Another characterisation, also valid for continuously differentiable functions, is given by the next Lemma. We recall that a {\em regular point} of $u$ is a point $z \in \SS^1$ such that, for all $t \in u^{-1}(\{z\})$, we have $u'(t) \neq 0$. The Sard's theorem states that if $u$ is continuously differentiable, then the set of nonregular points has measure $0$ in $\SS^1$. If $z$ is a regular point of $u$, and if $t \in u^{-1}(\{z\})$, we set 
\begin{equation} \label{eq:def:nuz}
    \nu_z[u, t] := {\rm sgn} \left( - \ri \frac{u'}{u}(t) \right) \quad \in \{-1, 1\}.
\end{equation}
If $\alpha(\cdot)$ is a lifting of $u$, then $\nu_z[u,t] = {\rm sgn} \left( \alpha'(t) \right)$. In other words, $\nu_z[u, t] = +1$ if $u(t)$ is locally turning positively, and $\nu_z[u,t] = -1$ if $u(t)$ is locally turning negatively.

\begin{lemma} \label{lem:winding_as_crossing}
    Let $u : \TT^1 \to \SS^1$ be continuously differentiable. If $z$ is a regular point of $u$, then $u^{-1}( \{z \})$ is finite, and
    \[
        W[u] = \sum_{t \in u^{-1}(\{z\})} \nu_z[u,t].
    \]
\end{lemma}
Although the proof is elementary, we provide it, as we use similar arguments in different contexts in the sequel.
\begin{proof}
     Let us first prove that $u^{-1} ( \{ z \})$ is finite. First, since for all $t \in u^{-1}( \{ z\})$, we have $u'(t) \neq 0$, the points in $u^{-1}( \{ z\} )$ are isolated. Assume by contradiction that $u^{-1} ( \{ z \})$ is infinite in the compact $\TT^1$. Then there is an accumulation point $t^* \in \TT^1$. By continuity of $u$, we must have $u(t^*) = z$, hence $t^* \in u^{-1} (\{ z\})$ as well. But this contradicts the fact that $t^*$ must be isolated. So $u^{-1} (\{ z\})$ is finite.
    
    Let $0 \le t_0 < t_1 < \cdots < t_{M-1} < 1$ be the pre-images of $z \in \SS^1$. Up to global rotations, we may assume without loss of generality that $z = 1 \in \SS^1$ and $t_0 = 0$, and we set $t_M = 1$. Let $\alpha$ be a continuously differentiable lifting of $u$. Since $u(t_m)= 1$, we must have $\alpha(t_m) = 2 \pi k_m$ with $k_m \in \Z$. We claim that 
    \begin{equation} \label{eq:k_versus_nu}
        k_{m+1} - k_m = \frac12 \left( \nu_1[u, t_{m}] + \nu_1[u, t_{m+1}]  \right).
    \end{equation}
    This would give the result, as
    \[
        W[u] = k_M - k_0 = \sum_{m=0}^{M-1} \left( k_{m+1} - k_m \right) = \frac12 \sum_{m=0}^{M-1} \left( \nu_1[u, t_{m}] + \nu_1[u, t_{m+1}] \right)  = \sum_{m=0}^{M-1} \nu_1[u, t_{m}].
    \]
    Let us prove~\eqref{eq:k_versus_nu} in the case $\nu_1[u, t_m] = 1$ and $\nu_1[u, t_{m+1}] = 1$ (the other cases are similar). By continuity of $\alpha'$, and since $\alpha'(t_m) > 0$ and $\alpha'(t_{m+1} ) > 0$, there is $0 < \varepsilon < \frac12 | t_{m+1} - t_m |$ so that 
    \[
        \alpha(t_m + \varepsilon) > \alpha(t_m) = 2 \pi k_m \quad \text{and} \quad
        \alpha(t_{m+1} - \varepsilon) < \alpha(t_{m+1}) = 2 \pi k_{m+1}.
    \] 
    The first inequality, together with the intermediate value theorem and the fact that $\alpha(t) \notin \Z$ for $t \in (t_m, t_{m+1})$ implies that $\alpha([t_m, t_{m+1}]) \subset 2 \pi [ k_m, k_m + 1]$. Similarly, the second inequality implies that $\alpha([t_m, t_{m+1}]) \subset 2  \pi [ k_{m+1} -1, k_{m+1}]$. By identification, this gives $k_{m+1} = k_m + 1$, and~\eqref{eq:k_versus_nu} is satisfied. The proof follows.
\end{proof}

If $u$ is not regular, or if we drop the signs of the crossings, we have a weak form of Lemma~\ref{lem:winding_as_crossing}.
\begin{lemma} \label{lem:winding_as_crossing_weak}
    Let $u$ be a continuous map from $\TT^1$ to $\SS^1$, and set $N := | W[u] | \in \N$. For all $z \in \SS^1$, there are at least $N$ points $0 \le t_0 < t_1 < \cdots < t_N < 1$ so that $u(t_k) = z$.
\end{lemma}
\begin{proof}
    Up to global rotations and symmetries, we may assume without loss of generality that $z = 1 \in \SS^1$, $u(0) = 1$, and $W[u] > 0$. Let $\alpha$ be a continuous lifting of $u$ with $\alpha(0) = 0$, so that $\alpha(1)= 2 \pi N$. With the intermediate value theorem, we deduce that for all $1 \le k \le N-1$, there is $t_k$ such that $\alpha(t_k) = 2 \pi k$. The result follows.
\end{proof}

\subsection{Ordinary Differential Equation}

\subsubsection{Notation for ODE}
\label{ssec:ODE}
In this section, we recall some classical facts about ODEs. We are mainly interested in the second order real-valued Hill's equation
\begin{equation*} 
    - u'' + V(x) u = E u \quad \text{on} \quad \R,
\end{equation*}
where $V \in L^1_\loc(\R)$ is some given real-valued locally integrable potential, and $E \in \R$. Later, $V$ will be either the periodic potential $V$, or the edge potential $V\chi  + (1 - \chi)V(\cdot - t)$. This setting allows to consider both cases at once, and handles the discontinuous case $\chi = \1( x < 0)$. Without loss of generality, we may absorb the energy $E$ into the potential, and study
\begin{equation} \label{eq:basicHilleqt2}
    - u'' + V(x) u = 0.
\end{equation}
We introduce the fundamental solutions $c_V$ and $s_V$, which are the solutions to the Cauchy problem~\eqref{eq:basicHilleqt2}, with the boundary conditions
\begin{equation} \label{eq:def:fundamental_solutions}
    c_V(0) = s_V'(0) = 1, \quad \text{and} \quad c_V'(0) = s_V(0) = 0.
\end{equation}
We first recall some basic facts. The proof is postponed until Section~\ref{ssec:proof_basicFact_ODE} for clarity.
\begin{lemma} \label{lem:BasicFact_ODE}
    For all $V \in L^1_{\rm loc}(\R)$, the functions $c_V$ and $s_V$ are well-defined, linearly independent and continuously differentiable on $\R$. The set of solutions to~\eqref{eq:basicHilleqt2} is the $2$-dimensional vectorial space
    \[
        \cL_V := \Span \{ c_V, s_V\}.
    \]
    If $u \in \cL_V$ is a solution to~\eqref{eq:basicHilleqt2}, then $u$ is continuously differentiable on $\R$, and
    \[
        \forall x \in \R, \quad u(x) = u(0) c_V(x) + u'(0) s_V(x).
    \]
    If in addition $u$ is non-null, then the zeros of $u$ are isolated and simple. For all $x \in \R$, $u(x)$ and $u'(x)$ cannot vanish at the same time.
\end{lemma}

Let $u \in \cL_V$ be a non-null solution of~\eqref{eq:basicHilleqt2}. Since $u$ and $u'$ cannot vanish at the same time, the complex-valued function
\begin{equation} \label{eq:def:theta}
    x \mapsto \theta[u, x] :=  \frac{u'(x) - \ri u(x)}{u'(x) + \ri u(x)} \quad \in \C,
\end{equation}
is well-defined, continuous, and has values in the unit circle $\SS^1 := \{ z \in \C, \ | z | = 1 \}$. We have $u(x) = 0$ iff $\theta[u,x] = 1$, while $u'(x) = 0$ iff $\theta[u,x] = -1$ (see Figure~\ref{fig:theta}).

\begin{figure}[H]
    \centering
    \begin{subfigure}{0.45\textwidth}
        \includegraphics[width=\textwidth]{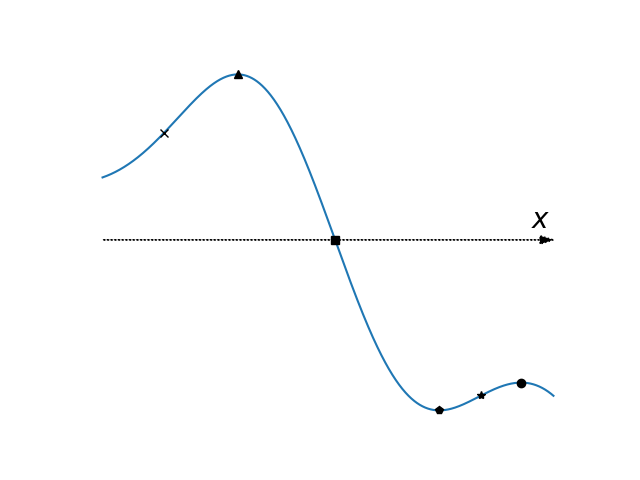}
        \caption{A solution $u(x)$.}
    \end{subfigure}
    \begin{subfigure}{0.45\textwidth}
        \includegraphics[width=\textwidth]{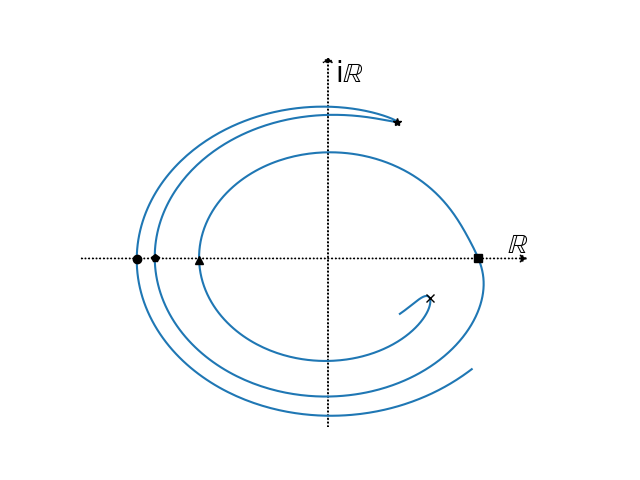}
        \caption{The corresponding $\theta[u,.]$. The radius is increasing with $x$ for clarity.}
    \end{subfigure}
    \caption{Sketch of a solution $u$ and the corresponding $x \mapsto \theta[u, x]$. We put some markers to track $x$ in the second picture.}
    \label{fig:theta}
\end{figure}

On the other hand, the zeros of $u$ are simple and isolated, so we can label them. We denote its set of zeros by
\[
    \cZ[u] :=\left( x_n \right)_{n \in \Z}, \quad \text{with} \quad \cdots < x_n < x_{n+1} < \cdots.
\]
If $u$ has a finite number of zeros, we put some of the $x_n$ to $\pm \infty$. Since $u$ does not vanish in the intervals $(x_n, x_{n+1})$, it has a constant sign on this interval, and this sign alternates between $(x_n, x_{n+1})$ and $(x_{n+1}, x_{n+2})$. The following remark is key to our analysis.

\begin{remark} \label{rem:theta_Vectu}
    For all $\lambda \in \R^*$, the function $\lambda u$ is another solution to~\eqref{eq:basicHilleqt2}, and we have $\theta[\lambda u, \cdot] = \theta[u, \cdot]$ and $\cZ[\lambda u, \cdot] = \cZ[u, \cdot]$. In other words, $\theta[u, \cdot]$ and $\cZ[u]$ only depends on $\Span \{ u \} \subset \cL_V$. \\
    If $u$ and $v$ are two non-null solutions to~\eqref{eq:basicHilleqt2}, then $u$ and $v$ are linearly dependent iff $\theta[u, x] = \theta[v, x]$ for all $x \in \R$, iff $\theta[u, x] = \theta[v, x]$ for some $x \in \R$, iff $\cZ[u] = \cZ[v]$.
\end{remark}


\subsubsection{The Maslov index for a periodic family of ODEs.}
\label{ssec:MaslovIndex}

We now consider a periodic family of potentials $\TT^1 \ni t \mapsto V_t \in L^1_\loc$. We say that this family is differentiable\footnote{We recall that $L^1_\loc(\R)$ is not a Banach space, but a complete metric space.} in $L^1_\loc(\R)$ if, for all $t \in \TT^1$, the function $(\partial_t V_t)$ is in $L^1_\loc(\R)$. This means that for all $K$ compact of $\R$,
\[
    t \mapsto \int_K V_t \quad \text{is differentiable, and} 
    \int_K \left( \partial_t V_t\right) := \partial_t \left( \int_K V_t \right).
\]
We set $\cL_{t} := \cL_{V_t}$, $c_t := c_{V_t}$ and $s_t := s_{V_t}$ for clarity. The proof of the next Lemma is postponed until Section~\ref{ssec:proof:ct_and_st_are_periodic}
\begin{lemma} \label{lem:ct_and_st_are_periodic}
    If $\left( V_t\right)_{t \in \TT^1}$ is a differentiable periodic family in $L^1_\loc(\R)$, then the functions $(t, x) \mapsto c_t(x)$ and $(t,x) \mapsto s_t(x) $ are continuously differentiable on $\TT^1 \times \R$.
\end{lemma}
Let $\left( \cL^+_t \right)_{t \in \TT^1}$ be a periodic family of vectorial spaces such that, for all $t \in \TT^1$, $\cL^+_t$ is a $1$-dimensional subspace of $\cL_t$. 
\begin{definition}[Continuity of vectorial spaces] \label{def:contL}
    We say that the map $t \mapsto \cL^+_t$ is continuously differentiable on $\TT^1$ if there are continuously differentiable functions $\lambda_c(\cdot)$ and $\lambda_s(\cdot)$ defined on $\R$ such that
    \begin{equation*}
    \cL^+_t = \Span \left\{  u_t(\cdot)  \right\}, \quad \text{where} \quad
    u_t := \lambda_c(t) c_t +   \lambda_s(t) s_t.
    \end{equation*}
\end{definition}
Although $\cL^+_t$ is periodic in $t$, we do not require $\lambda_c$ and $\lambda_s$ to be periodic {\em a priori}. According to Remark~\ref{rem:theta_Vectu}, since $\cL^+_t$ is $1$-dimensional, we can define 
\[
    \theta \left[ \cL^+_t, x \right] := \theta\left[ u_t, x\right], 
    \quad \text{and} \quad
    \cZ \left[  \cL^+_t \right] := \cZ[u_t ].
\]
We write $\theta_t^+$ and $\cZ_t^+$ for clarity. The function $(t, x) \mapsto \theta_t^+[x]$ is continuously differentiable from $\TT^1 \times \R$ to $\SS^1$. For all $x \in \R$, the map $t \mapsto \theta_t^+ [x]$ is periodic with value in $\SS^1$, hence has a winding number, and by continuity in $x$, this winding number is independent of the choice of $x$. 

\begin{definition}[Maslov index] \label{def:MaslovIndex}
    The Maslov index $\cM^+$ of $(\cL_t^+)_{t \in \TT^1}$ is the common winding number of $t \mapsto \theta_t^+[x] $.
\end{definition}

\begin{remark}
    As the next Lemma shows, this index can be interpreted as the number of intersections between the vectorial space $\cL^+_t$ and the Dirichlet vectorial space $u(x) = 0$, in the spirit of Maslov's work, hence the name.
\end{remark}

On the other hand, if $x_n \in \cZ_{t^*}$, we have $u_{t^*}(x_n) = 0$ and $u_{t^*}'(x_n) \neq 0$. So by the implicit function theorem, there is a continuously differentiable function $x_n(t)$ defined on a neighbourhood of $t^*$ with $x_n(t^*) = x_n$ such that $u_t(x_n(t)) = 0$ locally around $t^*$. We deduce that there are continuously differentiable functions $x_n(t)$ such that
\begin{equation} \label{eq:parametrisation_of_Z}
    \cZ_t^+ := \left\{\cdots < x_n(t) < x_{n+1}(t) < \cdots \right\}.
\end{equation}
By periodicity of $t \mapsto \cZ_t^+$, we have $\cZ_1^+ = \cZ_0^+$. We infer that there is $m \in \Z$ so that $x_n(1)= x_{n + m}(0)$ for all $n \in \Z$. The next result shows that the Maslov index can be seen as a flow of zeros.

\begin{lemma}[Characterisation of the Maslov index] \label{lem:MaslovIndex}
    The Maslov index is the integer $\cM^+$ so that $x_n(1) = x_{n+ \cM^+}(0)$.
\end{lemma}

\begin{proof}
    We want to use Lemma~\ref{lem:winding_as_crossing} with $z = 1 \in \SS^1$. To do so, we first need to prove that there is $x^* \in \R$ so that $z = 1 \in \SS^1$ is a regular point of $t \mapsto \theta_t^+[x^*]$.

    For $n \in \Z$, we denote by $X_n \subset \R$ the set of nonregular points for $x_n(\cdot)$. By Sard's theorem, $X_n$ has vanishing measure, hence so is $X := \cup_{n \in \N} X_n$. In particular, the set $\R \setminus X$ is non empty, and all $x \in \R \setminus X$ is a regular point for all the maps $x_n(\cdot)$. 
    
    Let $x^* \in \R \setminus X$ be such a point, and let $t^* \in \TT^1$ be such that $\theta_{t^*}^+[x^*] = 1$. From the definition of $\theta$ in~\eqref{eq:def:theta}, this implies that $u_{t^*}(x^*) = 0$, hence there is $n \in \Z$ so that $x_n(t^*)= x^*$. For all $t$ we have $u(t, x_n(t)) = 0$. Differentiating at $t = t^*$ gives
     \[
        x_n'(t^*) = - \dfrac{\partial_t u(t^*, x^*)}{\partial_x u(t^*, x^*)} ,
    \]
    and this number is not null, as $x^*$ is a regular point of $x_n(\cdot)$. On the other hand, differentiating $\theta_t^+[x^*]$ with respect to $t$ gives (we use that $u(t^*, x^*) = 0$)
    \[
        \partial_t \theta_t^+[x^*] = - 2 \ri \dfrac{\partial_t u(t^*, x^*)}{\partial_x u(t^*, x^*)} = 2 \ri x_n'(t^*) \quad \neq 0.
    \]
    We first deduce that $z = 1 \in \SS^1$ is indeed a regular point of $t \mapsto \theta_t^+[x^*]$. In addition, we see that $x_n(t)$ is locally moving forwards ($x_n'(t^*) > 0$) if and only if $\theta_t^+[x^*]$ is locally turning positively ($\nu_1[\theta_t, t^*] = 1$). The proof then follows from Lemma~\ref{lem:winding_as_crossing} and using arguments similar to its proof (see also Figure~\ref{fig:theta}).
\end{proof}


\section{Bulk-Edge for dislocations in the Schrödinger case}
\label{sec:HamiltonianCase}

In this section, we prove our results in the Schrödinger case. We consider $V$ a $1$-periodic function in $L^1_\loc(\R)$ (we write in the sequel $V \in L^1_\per(\R)$), and we define the operator
\[
    H_0 := - \partial_{xx}^2 + V, \quad \text{acting on $L^2(\R)$, with domain $H^2(\R)$,}
\]
Some references for the properties of such Hamiltonians and related ODEs are~\cite{reed1978analysis, poschel1987}. We provide the usual proofs of the following results in Appendix~\ref{appendix:TransfertMatrix} for completeness.
\begin{lemma}~\label{lem:Spectrum_Hill}
    The operator $H_0$ is self-adjoint, and its spectrum is purely essential: $\sigma(H_0) = \sigma_{\ess}(H_0)$. In addition, there is a sequence
    \[
        E_1^- < E_1^+ \le E_2^- < E_2^+ \le E_3^- < E_3^+ \le \dots
    \]
    such that $\sigma(H_0) = \bigcup_{n \in \N^*} [E_n^-, E_n^+]$.
\end{lemma}
The interval $[E_n^-, E_{n}^+]$ is called the $n$-th band of $H$, and $g_n := (E_n^+, E_{n+1}^-)$ is the $n$-th gap (with $g_0 := (-\infty, E_1^-)$). If $E_n^+ = E_{n+1}^-$, the gap is empty, and it is open otherwise.

\medskip

Closely related to the operator $H_0$ is the ODE $-u'' + Vu = Eu$. We need the following result.
\begin{lemma} \label{lem:LEpm}
    For all $E \notin \sigma(H_0)$, the vectorial space $\cL(E)$ of solutions to $-u'' + Vu = Eu$ has a splitting $\cL(E) = \cL^+(E) \oplus  \cL^-(E)$, where $\cL^{\pm}(E)$ is the vectorial space of solutions that are square integrable at $\pm \infty$. In addition:
    \begin{itemize}
        \item The spaces $\cL^\pm(E)$ are of dimension $1$, and $\cL^+(E) \cap \cL^-(E) = \emptyset$;
        \item The solutions in $\cL^\pm(E)$ are exponentially decaying at $\pm \infty$;
        \item The maps $E \mapsto \cL^\pm(E)$ are differentiable, in the sense of Definition~\ref{def:contL}.
    \end{itemize}
\end{lemma}

\subsection{The bulk index for translated Hill's operators}
\label{ssec:bulk_index}

We now define our bulk index for the $n$-th gap. We denote by $W^{1,1}_\per(\R)$ the Sobolev space of periodic distributions $V$ such that $V$ and $V'$ are in $L^1_\per(\R)$. The following Lemma is straightforward.
\begin{lemma}\label{lem:VinW1per}
    The map $t \mapsto V(\cdot - t)$ is differentiable in $L^1_\loc(\R)$ iff $V \in W^{1,1}_\per(\R)$. 
\end{lemma}

We fix $V \in W^{1,1}_\per(\R)$ and $H_0 := - \partial_{xx}^2+ V$. For $t \in \TT^1$, we set $V_t(x) := V(x-t)$, and
\[
    H(t) := - \partial_{xx}^2 + V_t, \quad \text{acting on $L^2(\R)$ with domain $H^2(\R)$}.
\]
Since $V$ is $1$-periodic, we have $V_{t+1} = V_t$ and $H(t+1) = H(t)$, so $H(\cdot)$ is a $1$-periodic family of operators.

\medskip

Let $\tau_t f(x) := f(x-t)$ be the translation operator acting on $L^2(\R)$. The operator $\tau_t$ is unitary, with $\tau_t^{-1} = \tau_t^* = \tau_{-t}$, and we have $H(t) = \tau_t H_0 \tau_{t}^*$. In particular, the spectrum $\sigma(H(t))$ is independent of $t$, and equals the one of $H_0$ given by Lemma~\ref{lem:Spectrum_Hill}. 

\medskip

Let $E$ be in the $n$-th gap, that we suppose open. The vectorial space $\cL_t(E)$ of solutions of $-u'' + V_t u = E u$ satisfies $\cL_t(E) = \tau_t \cL(E)$, and, by Lemma~\ref{lem:LEpm}, it has a natural splitting  $\cL_t(E) = \cL_t^+(E) \oplus \cL_t^-(E)$ with $\cL_t^\pm(E) = \tau_t \cL^\pm(E)$. The functions in $\cL^\pm(E)$ are the solutions which are exponentially decaying at $\pm \infty$.

The maps $t \mapsto \cL_t(E)$ is periodic in $t$. Also, both $\cL_t^-(E)$ and $\cL_t^+(E)$ are $1$-dimensional.
\begin{lemma} \label{lem:continuityOfL}
    For all $E \in g_n$, the maps $t \mapsto \cL^{\pm}_t(E)$ are differentiable on $\TT^1 \times g_n$, in the sense of Definition~\ref{def:contL}.
\end{lemma}
See Section~\ref{ssec:proof:lem_continuityOfL} for the proof. According to Section~\ref{ssec:MaslovIndex}, we can attach a Maslov index to the families $\left(\cL_t^\pm(E) \right)_{t \in \TT^1}$. By continuity in $E$, these indices are independent of $E \in g_n$. We denote by $\cM^\pm$ these indices.

\begin{lemma} \label{lem:M+=M-}
    $\cM^+ = \cM^-$.
\end{lemma}

\begin{proof}
    We define $\Omega_t(E, x) := \theta[\cL^+_t(E),x] \left( \theta[\cL^-_t(E),x] \right)^{-1} \in \SS^1$. Since $\cL_t^+(E) \cap \cL_t^-(E) = \emptyset$ for all $t \in \TT^1$, we have $\theta^+_t(x) \neq \theta^-_t(x)$ for all $t \in \TT^1$. In particular, $\Omega_t(E, x)$ is never equal to $1 \in \SS^1$. By the reciprocal of Lemma~\ref{lem:winding_as_crossing_weak}, we deduce that its winding is null. Together with the fact that the winding number is a group homomorphism, we obtain
    \[
        0 = W [t \mapsto \Omega_t(E, x)] = W[ t \mapsto \theta[\cL^+_t(E),x]] - W[t \mapsto  \theta[\cL^-_t(E),x]] = \cM^+ - \cM^-.
    \]
\end{proof}

\begin{definition}[Bulk index] \label{def:bulkIndex}
    We define the {\em bulk index} $\cB_n$ as the common Maslov index $\cB_n = \cM^+ = \cM^-$.
\end{definition}

\begin{proposition} \label{prop:bulk=n}
    If the $n$-th gap is open, then $\cB_n = n$.
\end{proposition}

\begin{proof}
    Let $u$ be a solution in $\cL^+(t = 0,E)$. We set $u_t(x) := u(x-t)$, so that $u_t \in \cL_t^+(E)$ for all $t \in \R$, and $\cZ(t) := \cZ_{u_t}$. If $\left\{ \cdots < x_n < x_{n+1} < \cdots \right\}$ is a parametrisation of $\cZ(0)$, then 
    \[
        \cZ(t) = \left\{ \cdots < x_n + t < x_{n+1} + t < \cdots \right\}.
    \]
    Up to global translation, we may assume $x_0 = 0$. From Lemma~\ref{lem:MaslovIndex}, $\cB_n$ is the integer so that $x_{\cB_n} = x_0 +1 = 1$. Hence $\cB_n$ is also the number of zeros of $u$ in the interval\footnote{We prove in Appendix~\ref{appendix:TransfertMatrix} that there is $\lambda \in (-1, 1)$ such that $u(x + 1) = \lambda u(x)$. So if $x \in \cZ(t)$, then $x +1 \in \cZ(t)$ as well, and $\cB_n$ is the number of zeros of $u$ in any interval of the form $[x, x+1)$.} $[0, 1)$.
    
    \medskip
    
    Since $u$ vanishes at $x = 0$ and $x = 1$, $u |_{[0, 1]}$ is an eigenvector (corresponding to the eigenvalue $E$) of the Dirichlet operator
    \[
        H^D := - \partial_{xx}^2 + V \quad \text{acting on $L^2([0,1])$ with domain $H^2_0([0, 1])$}.
    \]
    It is a well-known fact (we provide a proof in Appendix~\ref{sec:appendix:Spectra} for completeness) that the spectrum of $H^D$ is discrete, composed of simple eigenvalues $\delta_1 < \delta_2 < \cdots$, that $\delta_n$ is the only eigenvalue of $H^D$ in the $n$-th gap of $H_0$, and that its corresponding eigenvector vanishes $n$ times in the interval $[x_n, x_n + 1)$. We deduce that $E = \delta_n$, and $u$ is the corresponding eigenvector. The proof follows.
\end{proof}

\subsection{The Chern number}
\label{ssec:Chern}

There is another natural bulk index that we can define, which corresponds to a Chern number. 

For all $t \in \TT^1$, the operator $H(t)$ commutes with $\Z$-translations, hence can be Bloch decomposed (see {\em e.g.}~\cite[Chapter XIII.16]{reed1978analysis} for instance). For $k \in \R$, we denote by $H(t, k)$ the Bloch fibers
\[
H(t, k) := - \partial_{xx}^2 + V_t \quad \text{acting on $L^2([0,1])$, with domain $H^2_k$},
\]
where we introduced (the normalisation here differs from the usual one)
\begin{equation} \label{eq:def:H2k}
H^2_k := \left\{ u \in H^2_\loc(\R), \ u(x +1) = \re^{ 2 \ri \pi k} u(x)  \right\}.
\end{equation}
The spaces $H^2_k$ are $1$-periodic in $k$, so the operators $H(t, k)$ are periodic in both $t$ and $k$. In the sequel, we write $\TT^2 := \TT^1 \times \TT^1$ for the $2$-torus where $(t,k)$ lives. If the $n$-th gap is open, and if $E$ lies into this gap, we can define the spectral projector
\[
    P_n(t,k) := \1 \left( H(t,k) \le E \right) \quad \text{acting on} \quad L^2([0,1]).
\]
This is the rank-$n$ projector on the eigenvectors corresponding to the $n$ lowest eigenvalues of $H(t,k)$. In particular, $\Ran \, P_n(t, k) \subset H^2_k$. The proof of the next Lemma is postponed until Section~\ref{ssec:proof:regularityPn}. We denote by $\cB(L^2([0,1]))$ the Banach space of bounded operators acting on $L^2([0,1])$.
\begin{lemma} \label{lem:regularityPn}
    Assume $W \in W^{1,1}_\per(\R)$. Then $(t, k) \to P_n(t, k)$ is continuously differentiable from $\TT^2$ to $\cB(L^2([0,1]))$, that is $\| \partial_{t} P_n \|_{\cB(L^2([0,1]))} +  \| \partial_{k} P_n \|_{\cB(L^2([0,1]))} < \infty$ for all $(t,k) \in \TT^2$.
\end{lemma}

For such family of operators, we can define a {\rm Chern number}\footnote{Since $V$ is real-valued, we have $P(t, -k) = K P(t,k) K$, where $K$ is the complex conjugation operator. This is different from time-reversal symmetry~\cite{panati2007triviality}, which is of the form $P(-t, -k) = K P(t,k) K$.}, which is given by the integral
\begin{equation} \label{eq:def:Chern}
    \Ch \left( P_n \right) := \dfrac{1}{2 \ri \pi} \iint_{\TT^2}  \Tr_{L^2([0,1])} \left( P_n \rd P_n \wedge \rd P_n  \right) \qquad \in \Z.
\end{equation}
Let us give another characterisation of this Chern number (and prove that it is indeed integer valued) using the notion of frames. This was already used in {\em e.g.}~\cite{cances2017robust, cornean2019localised}. 

\medskip

We say that a family of $n$-vectors $\Psi(t,k) = \left( \psi_1, \psi_2, \cdots \psi_n \right)(t,k) \in \left(L^2([0,1]) \right)^n$ is a (continuously differentiable) frame for $P_n$ on the cut torus $\dot \TT^2 := [0,1] \times \TT^1$ if
\begin{itemize}
    \item For all $(t,k)$ in $\dot \TT^2$, we have $\bra \psi_i(t,k), \psi_j(t,k) \ket = \delta_{i,j}$, which we write $\Psi^*(t,k) \Psi(t,k) = \bbI_n$.
    \item For all $(t,k) \in \dot \TT^2$, we have $P_n(t,k) = \sum_{i=1}^n | \psi_i(t,k) \ket \bra \psi_i(t,k) |$ which we write $P(t,k) = \Psi(t,k) \Psi^*(t,k)$.
    \item The map $(t,k) \mapsto \Psi(t,k)$ is continuously differentiable on $\dot \TT^2$.
\end{itemize}
The first condition ensures that the family $\Psi$ is an orthonormal family of $n$ vectors, and the second condition states that this family span the range of $P_n$. By periodicity of $t \mapsto P_n(t, k)$, we get that for all $k \in \TT^1$, the families $\Psi(t = 1,k)$ and $\Psi(t = 0,k)$ both span the range of $P_n(0, k)$, hence there is a unitary $U(k) \in \U(n)$ such that
\begin{equation} \label{eq:def:obstructionMatrix}
    \forall k \in \TT^1, \quad \Psi(0,k)= \Psi(1, k) U(k), 
    \quad \text{which gives} \quad
    U(k) = \Psi(1,k)^* \Psi(0,k).
\end{equation}
Since both $\Psi(0,k)$ and $\Psi(1,k)$ are periodic in $k$, then so is $U(k)$. In particular, $k \mapsto \det U(k)$ is a map from $\TT^1$ to $\SS^1$.
\begin{lemma}
    The Chern number of $P_n$ defined in~\eqref{eq:def:Chern} equals the winding number of $k \mapsto \det U(k)$. In particular, the latter quantity is independent of the frame.
\end{lemma}
\begin{proof}
    This result was proved in details in~\cite{cornean2019localised}. We recall the main key steps for completeness. Since $P_n = \Psi \Psi^*$, and $\Psi^* \Psi = \bbI_n$, we have using the cyclicity of the trace,
    \begin{align*}
        & \Tr \left( P_n \rd P_n \wedge \rd P_n \right)  =
        \Tr(\Psi \Psi^* \left[(\rd \Psi) \Psi^* + \Psi (\rd \Psi^*) \right] \wedge  \left[(\rd \Psi) \Psi^* + \Psi (\rd \Psi^*) \right] ) \\
             & \quad = \Tr \left( \Psi^* (\rd \Psi) \Psi^* \wedge (\rd \Psi)   \right) + 
                \Tr \left( \Psi \Psi^* (\rd \Psi) \wedge (\rd \Psi^*)   \right) +
                \Tr \left( \rd \Psi^* \wedge \rd \Psi \right) +
                \Tr \left( \Psi (\rd \Psi^*) \wedge \Psi(\rd \Psi^*) \right).
    \end{align*}
    Using that $\Tr( f \rd g \wedge f \rd g) = 0$ by anti-symmetry of the $2$-form $f \rd g \wedge f \rd g$, and the fact that $\Psi^* \Psi = \bbI_n$, so that $\rd \Psi^* \Psi + \Psi^* \rd \Psi = 0$, we obtain (see also~\cite{simon1983holonomy})
    \[
        \Tr \left( P_n \rd P_n \wedge \rd P_n \right)  = \rd \Tr \left( \Psi^* \rd \Psi \right).
    \]
    We now apply Stokes' theorem in~\eqref{eq:def:Chern} on the cut torus $\dot \TT^2$. This gives
    \[
        \Ch (P_n) = \frac{1}{2 \ri \pi} \iint_{\dot \TT^2} \Tr \left( P_n \rd P_n \wedge \rd P_n \right) 
        = \dfrac{1}{2 \ri \pi} \int_0^{1} \left[ \Tr \left( \Psi^* \partial_k \Psi \right)(0, k)  -  \Tr \left( \Psi^* \partial_k \Psi \right)(1, k)  \right] \rd k.
    \]
    On the other hand, differentiating~\eqref{eq:def:obstructionMatrix} and using that $(\partial_k \Psi^*) \Psi + \Psi^* (\partial_k \Psi) = 0$, we get
    \[
         \Tr \left(U^* \partial_k U \right). =  \Tr \left( \Psi^* \partial_k \Psi \right)(0, k)  -  \Tr \left( \Psi^* \partial_k \Psi \right)(1, k).
    \]
    Altogether, we obtain
    \[
         \Ch (P_n) = \dfrac{1}{2 \ri \pi} \int_0^{1} \Tr \left(U^* \partial_k U \right) \rd k
         =  \dfrac{1}{2 \ri \pi} \int_0^{1} \dfrac{ \partial_k \left( \det U(k) \right) }{\det U(k) } \rd k,
    \]
    which is the winding number of $\det U$, as wanted.
\end{proof}

\begin{proposition} \label{prop:Chern}
    If the $n$-th gap is open, then $\Ch \left( P_n \right) = n$.
\end{proposition}
\begin{proof}
For $ t= 0$, we consider a continuous and periodic frame $\Psi(0, k)$ for $P$ on the $1$-torus $\{ 0 \} \times \TT^1 \subset \dot \TT^2$. Since $\Psi \subset \Ran(P_n) \subset H^2_k$, we have the quasi-periodic condition
\begin{equation} \label{eq:psij_quasiperiodic}
\Psi(t=0, k, x+1) = \re^{2 \ri \pi k} \Psi(t=0,k, x).
\end{equation}
Using that $P(t,k) = \tau_t P(0,k) \tau_t^*$, we can extend continuously this frame on the whole cut torus $\dot \TT^2$ with
\[
    \forall t, k \in \dot \TT^2, \quad \Psi(t,k,x) := \Psi(0, k,x-t).
\]
To find the mismatch between the frame at $t = 0$ and $t = 1$, we use~\eqref{eq:psij_quasiperiodic}, and find that
\[
    \Psi(1, k,x) = \Psi(0, k,x - 1) = \re^{-2 \ri \pi k} \Psi(0, k,x).
\]
By identification with~\eqref{eq:def:obstructionMatrix}, we obtain $U(k) = \re^{2 \ri \pi k} \bbI_n$, whose determinant has winding $n$.
\end{proof}

\subsection{Edge modes}
\label{ssec:edgeIndex}

We now focus on the self-adjoint edge Schrödinger operator
\[
    H_\chi^\sharp(t) := H_0 \chi + H(t) (1 - \chi) =
    - \partial_{xx}^2 + V \chi + V_t (1 - \chi), \quad \text{acting on $L^2(\R)$, with domain $H^2(\R)$}.
\]
Here, $\chi$ is a switch function, that is an $L^\infty(\R)$ function satisfying $\chi(x) = 1$ for $x < -L$, and $\chi(x) = 0$ for $x > L$, where $L$ is any fixed number. We set $V_t^\chi(x) := V \chi + V_t (1 - \chi)$ in the sequel. The following result is straightforward.

\begin{lemma}
    For all $V \in W^{1,1}_\per(\R)$ and $\chi \in L^\infty(\R)$, the map $t \mapsto V_t^\chi$ is differentiable in $L^1_\loc(\R)$.
\end{lemma}

 It is classical that the essential spectrum of $H_\chi^\sharp(t)$ is 
\[
    \sigma_\ess(H_\chi^\sharp(t)) = \sigma_\ess (H_0) \cup \sigma_\ess \left( H(t) \right) = \sigma (H_0) = \bigcup_{n \in \N^*} [E_n^-, E_n^+].
\]
Let $E \in g_n$ be in the $n$-th essential gap, that we suppose open. The ODE equation 
\begin{equation} \label{eq:edgeODE}
    -u'' +  V^\chi_t u  = Eu
\end{equation}
can be studied with the tools developed in Section~\ref{ssec:ODE}. We denote by $\cL^{\sharp,+}_{\chi, t}(E)$ and $\cL^{\sharp, -}_{\chi,t}(E)$ the (edge) vectorial spaces of solutions that are square integrable at $+ \infty$ and $-\infty$ respectively. They differ from the (bulk) ones $\cL^\pm_t(E)$ introduced in Lemma~\ref{lem:LEpm}. In this lemma, we proved that $\cL^+_t(E)$ and $\cL^-_t(E)$ are always disjoint. This is different for the edge ones, and the spaces $\cL^{\sharp,+}_{\chi, t}(E)$ may have a non trivial intersection. If this happens, the elements $u \in \cL^{\sharp, +}_{\chi,t}(E) \cap \cL^{\sharp, -}_{\chi,t}(E)$ are normalisable, hence are eigenvectors of $H_\chi^\sharp(t)$ for the eigenvalue $E$. We call such elements {\em edge modes}, or {\em edge states}.

The maps $t \mapsto \cL^{\sharp, \pm}_{\chi, t}(E)$ are $1$-periodic and continuously differentiable. We prove below that they are of dimension $1$. Hence, the functions $\theta^{\sharp, \pm}_{\chi, t}(E, x) := \theta[ \cL^{\sharp, \pm}_{\chi, t}(E), x]$ are well-defined, and we can introduce the edge quantity
\begin{equation} \label{eq:def:Omega}
    \Omega^\sharp_\chi (t, E, x) := \theta^{\sharp, +}_{\chi, t} \left(E , x  \right) \left(  \theta^{\sharp, -}_{\chi, t} \left(E , x  \right) \right)^{-1}.
\end{equation}
As a function of $t$, it is a map from $\TT^1$ to $\SS^1$, and $\Omega^\sharp_\chi (t, E, x) = 1$ iff $\theta^{\sharp, +}_{\chi, t} \left(E , x  \right)  = \theta^{\sharp, -}_{\chi, t} \left(E , x  \right)$. The following lemma is straightforward from the previous discussion.
\begin{lemma} \label{lem:spectrum_Hedge}
    The following assertions are equivalent:
    \begin{enumerate}[(i)]
        \item $E$ is an eigenvalue of $H^{\sharp}_\chi(t)$;
        \item $\cL^{\sharp, +}_{\chi, t}(E) = \cL^{\sharp, -}_{\chi, t}(E)$;
        \item there is $x \in \R$ such that $\Omega^\sharp_\chi (t, E, x) =1$;
        \item for all $x \in \R$, we have $\Omega^\sharp_\chi (t, E, x) = 1$.
    \end{enumerate}
    If these assertions are satisfied, then
    $
    {\rm Ker} \left( H^{\sharp}_\chi(t)  - E \right) = \cL^{\sharp, +}_{\chi, t}(E) = \cL^{\sharp, -}_{\chi, t}(E).
    $
    In particular, all eigenvalues of $H^\sharp_\chi(t)$ are simple, and all edge modes are exponentially decaying at $\pm \infty$.
\end{lemma}

This motivates the following definition.
\begin{definition}[Edge index] \label{def:edgeIndex}
    We define the edge index $\cI^{\sharp}_{\chi, n}$ as the winding $W \left[  \Omega^\sharp_\chi (\cdot , E, x) \right]$.
\end{definition}
By continuity, this index is independent of $E \in g_n$ and $x \in \R$.

\begin{proposition} \label{prop:edge=n}
    If the $n$-th gap is open, then $\cI^\sharp_{\chi, n} = n$. In particular, it is independent of~$\chi$.
\end{proposition}

\begin{proof}
    Since the winding number is a group homomorphism, we have
    \[
        \cI^\sharp_{\chi, n} = W \left[  \Omega^\sharp_\chi (\cdot , E, x) \right] = \cM^{\sharp, +}_{\chi} - \cM^{\sharp, -}_{\chi},
        \quad \text{where} \quad 
         \cM^{\sharp, \pm}_{\chi} := W \left[  t \mapsto \theta^{\sharp, \pm}_{\chi, t} \left(E , x  \right)  \right].
    \]
    Let us compute $\cM^{\sharp, +}_{\chi}$. For $x > L$, the solutions $u \in \cL^{\sharp, +}_{\chi, t}(E)$ decay at $+ \infty$, and satisfies the bulk-like equation
    \[
        \forall x > L, \quad -u'' + V_t u  = E u.
    \]
    Together with Lemma~\ref{lem:LEpm}, we deduce that $u$ is actually exponentially decaying at $+ \infty$, and that, for $x > L$, we have $\theta[ \cL^{\sharp, +}_{\chi, t}(E), x ] = \theta[ \cL^+_{t}(E), x ]$. In particular, $\cL^{\sharp, +}_{\chi, t}(E)$ is indeed of dimension $1$. Also, with the Definition~\ref{def:MaslovIndex} of the Maslov index applied with $x > L$, we deduce that 
    \[
        \cM^{\sharp, +}_{\chi} = W \left[t \mapsto \theta[ \cL^{\sharp, +}_{\chi, t}(E), x ] \right] 
        = W \left[t \mapsto \theta[ \cL^{+}_{t}(E), x ] \right]
        = \cM^+ = \cB_n.
    \]
   Similarly, for $x < -L$, the solutions $u \in \cL^{\sharp, -}_{\chi, t}(E)$ decay at $-\infty$ and satisfy the equation
    \begin{equation*}
        \forall x < -L, \quad -u'' + V u  = Eu,
    \end{equation*}
    which is independent of $t$ and $\chi$. In particular $\theta[\cL^{\sharp, -}_{\chi, t}(E), x] = \theta[\cL^-_{t = 0}(E),x]$ is also independent of $t$, hence has null winding number. This proves that $\cM^{\sharp, -}_{\chi} = 0$, and that $\cL^{\sharp, -}_{\chi, t}$ is of dimension $1$.
    Finally, together with Proposition~\ref{prop:bulk=n}, we obtain as wanted
    \[
        \cI^\sharp_{\chi, n} = \cM^{\sharp, +}_{\chi} - \cM^{\sharp, -}_{\chi} = \cB_n  - 0 = n.
    \]
\end{proof}
Proposition~\ref{prop:edge=n} and Lemma \ref{lem:spectrum_Hedge} already imply the existence of edge modes.
\begin{lemma}[Existence of edge modes] \label{lem:existenceEdgeModes}
    Let $E$ be in the $n$-th essential gap, that we suppose open. Then there are at least $n$ values $0 < t_1 < t_2 < \cdots < t_n < 1$ so that $E$ is an eigenvalue of $H^\sharp_\chi(t_k)$.
\end{lemma}

\begin{proof}
    From~\eqref{eq:def:Omega}, the winding number of $t \mapsto \Omega^\sharp_\chi (\cdot,  E, x)$ is
    \[
        W \left[ \Omega^\sharp_\chi (\cdot,  E, x) \right] = \cM^{\sharp, +}_{\chi,n} - \cM^{\sharp, -}_{\chi,n}  = n.
    \]
    Together with Lemma~\ref{lem:winding_as_crossing_weak} with $z = 1$, we deduce that there are at least $n$ points $0 \le t_0 < \cdots < t_n < 1$ so that $\Omega^\sharp_\chi (t_k, E, x) = 1$. At these points, $E$ is an eigenvalue of $H^\sharp_\chi(t_k)$. Finally, since $H^\sharp_\chi(t = 0)= H_0$ has only purely essential spectrum, we must have $t_0 > 0$.
\end{proof}
As we see from the proof, if we attach an orientation to each regular crossing $t_k$, then we could use Lemma~\ref{lem:winding_as_crossing} instead of Lemma~\ref{lem:winding_as_crossing_weak} in the previous result, and have a finer result. This is detailed in the next Section.

\subsection{The domain wall spectral flow}
\label{ssec:spectralFlowDW}

In the previous section, we fixed the energy $E$. We now investigate how the spectrum varies with $t$. We introduce in this section the notion of spectral flow. Our approach is slightly different than the usual one (see {\em e.g.}~\cite{phillips1996self}), but is equivalent, up to a global sign, and easier to manipulate for our purpose. First, we recall the following classical result.
\begin{lemma} \label{lem:continuityOfEigenvalues}
    Let $E^*$ in the $n$-th essential gap be an eigenvalue of $H_\chi^\sharp(t^*)$. There is $0 < t^- < t^* < t^+ < 1$ and a continuously differentiable map $t \in (t^-, t^+) \mapsto E(t^*)$ with $E(t_0) = E^*$ such that
    \begin{itemize}
        \item For all $t \in (t^-, t^+)$, $E(t)$ is an eigenvalue of $H_\chi^\sharp(t)$;
        \item $\lim_{t \to t^-} E(t)$ and $\lim_{t \to t^+} E(t)$ belong to the band edges $\{ E_n^+, E_{n+1}^-\}$.
    \end{itemize}
In addition, the triplet $(t^-, t^+, E(\cdot))$ is unique.
\end{lemma}
We postpone the proof until Section~\ref{ssec:proof:continuityOfEigenvalues}, and just highlight the fact that the uniqueness property comes from the simplicity of the eigenvalues of $H_\chi^\sharp(t_0)$ (see Lemma~\ref{lem:spectrum_Hedge}). In the sequel, we say that the spectrum of $t \mapsto H_\chi^\sharp(t)$ is continuous, and we call such function $E(\cdot)$ a branch of eigenvalues.

\begin{lemma} \label{lem:countableEigenvalues}
    There is a countable number of branches of eigenvalue $\left( E_k(\cdot) \right)_{k \in K}$ with $K \subset \Z$ such that, 
    \[
         \forall t \in \TT^1, \quad \sigma \left( H^\sharp_\chi(t) \right) \cap g_n = \bigcup_{k \in K} \left\{ E_k(t) \right\}.
    \]
\end{lemma}

\begin{proof}
    For $m \in \N^*$, we introduce the energy interval $I_m := (E_n^+ + \frac{1}{m}, E_{n+1}^- - \frac{1}{m})$. For $m$ large enough, $I_m$ is non empty, and its closure $\overline{I_m}$ is included in $g_n$. Also, $\cup_{m \in \N} I_m = g_n$. For all $t \in \TT^1$, the spectrum of $H^\sharp_\chi(t)$ in $g_n$ is composed of eigenvalues that can only accumulate at the band edges $\{E_n^+, E_{n+1}^- \}$, so is finite in $I_m$. Since $\TT^1$ is compact, and by continuity of the spectrum,
    \[
    \sup_{t \in \TT^1} \left[ {\rm Card} \left( \sigma \left( H^\sharp_\chi(t) \right) \cap I_m \right) \right] < \infty.
    \]
    This means in particular that we can find a finite number of branches of eigenvalues $\left( E_k^m(\cdot) \right)_{1 \le k \le K_m}$ so that
    \[
    \forall t \in \TT^1, \quad \sigma \left( H^\sharp_\chi(t) \right) \cap I_m = \bigcup_{1 \le k \le K_m} \left\{  E_k^m(t) \right\}.
    \]
    The result follows by considering the union in $m \in \N^*$.
\end{proof}

We say an energy $E \in g_n$ is {\em regular for the spectrum of} $H^\sharp_\chi(\cdot)$ if it is regular point for all branches $\left( E_k(\cdot) \right)_{k \in K}$. Let $X_k \subset g_n$ be the set of non-regular points of $E_k(\cdot)$. By Sard's theorem, $X_k$ is of measure $0$, hence so is the union $X := \bigcup_{k \in K} X_k$. Any point $E \in g_n \setminus X$ is a regular point for the spectrum of $H^\sharp_\chi(\cdot)$.

\begin{definition}[Spectral flow] \label{def:spectralFlow}Let $E \in g_n$ be a regular point for the spectrum of $H^\sharp_\chi(\cdot)$. The spectral flow of $t \mapsto H^\sharp_\chi(t)$ in the $n$-th gap is
    \[
    \cS^\sharp_{\chi, n} := - \sum_{k \in K} \sum_{t \in E_k^{-1} (\{ E \})} {\rm sgn} \left( E_k'(t)  \right).
    \]
\end{definition}
From the proof of Lemma~\ref{lem:countableEigenvalues}, there are only a finite number of branches that can touch the energy $E$, so the sum in $k \in K$ is actually finite.
\begin{remark}
    The convention here differs from the usual one by a global minus sign. This is because in our case, the eigenvalues moves from the upper band to the lower one. 
\end{remark}

Using similar arguments to the ones in the proof of Lemma~\ref{lem:winding_as_crossing}, one can see that the spectral flow counts the net number of eigenvalues going from the upper band to the lower one, as $t$ goes from $0$ to $1$. The fact that the spectral flow is independent of the choice of the regular point $E$ is a consequence of the following result.

\begin{proposition} \label{prop:Schi=n}
        If the $n$-th essential gap is open, then $\cS^\sharp_{\chi, n}  = \cI_{\chi,n}^\sharp = n$. In particular, it is independent of~$\chi$ (and $E$).
\end{proposition}

\begin{proof}
   We fix $x_0 < -L$ (this choice simplifies the computations below). We claim that $1 \in \SS^1$ is a regular point of $\Omega^\sharp_\chi(\cdot, E, x_0)$ iff $E$ is a regular point for the spectrum of $H_\chi^\sharp(\cdot)$. Let $t^*, E^*$ be such that $\Omega^\sharp_\chi(t^*, E^*, x_0) = 1$. By Lemma~\ref{lem:spectrum_Hedge}, this implies that $E^*$ is an eigenvalue of $H^\sharp_\chi(t^*)$, hence there is $k \in K$ so that $E^* = E_k(t^*)$.
    
    Our goal is to prove that 
    \begin{equation} \label{eq:IPPequality}
        \nu_{z = 1}  \left[  \Omega^\sharp_\chi (\cdot, E^*, x_0), t^*  \right] = - {\rm sgn} (E_k'(t^*) ),
    \end{equation}
    where we recall that $\nu_z[\cdot, \cdot]$ was defined in~\eqref{eq:def:nuz}. This would first prove that $1 \in \SS^1$ is indeed a regular point of $\Omega^\sharp_\chi(\cdot, E, x_0)$, and it would also prove the result since $\cI^{\sharp}_{\chi, n}$ is the winding of $\Omega^\sharp_\chi$.
    Let us first compute $E_k'(t^*)$ (which is non null by assumption). Let $u_t(\cdot) \in L^2(\R)$ be a normalised eigenvalue of $H^\sharp_\chi(t)$ for the eigenvalue $E(t)$. From the Hellman-Feynman theorem, we have
    \begin{align} \label{eq:HellmanFeynman}
         E_k' (t) & = \partial_t \left\bra u_t, H^\sharp_\chi(t) u_t \right\ket 
            = \left\bra u_t, \partial_t \left[ H^\sharp_\chi(t) \right] u_t \right\ket
            = \left\bra u_t,  (\partial_t V^\chi_t) u_t \right\ket = \int_{\R} \left( \partial_t V^\chi_t \right)  | u_t |^2 .
    \end{align}
    We now compute $\nu_{z = 1} \left[  \Omega^\sharp_\chi (\cdot, E^*, x_0), t^*  \right]$. From the definition of $\Omega^\sharp_\chi$, and the fact that $\theta \left[ \cL^{\sharp, -}_{\chi, t}(E^*), x_0  \right]$ is independent of $t$ (since $x_0 < -L$), we obtain
    \begin{equation} \label{eq:fraction}
       - \ri \dfrac{ \left( \partial_t \Omega^\sharp_\chi \right) (t^*, E^*, x_0) }{ \Omega^\sharp_\chi (t^*, E^*, x_0)} 
       =  - \ri \dfrac{ \left( \partial_t \theta \left[ \cL_{\chi, t}^{\sharp, +}(E^*), x_0 \right] \right)}{\theta \left[ \cL_{\chi, t}^{\sharp, +}(E^*), x_0 \right]} \Big|_{t = t^*}
       = 2 \dfrac{u_+ ( \partial_t u_+')-  u_+' (\partial_t u_+)}{| u_+ |^2 + | u_+'|^2} \Big|_{t = t^*, x= x_0},
    \end{equation}
    where $u_+(t, x)$ is any continuously differentiables branch of functions in $\cL^{\sharp, +}_{\chi,t}(E^*)$. In particular, for all $t \in \TT^1$, $u_+(t, \cdot)$ satisfies the ODE $( - \partial_{xx}^2 + V^\chi_t - E^*) u_+ = 0$.
    Differentiating with respect to $t$ gives
    \[
        ( - \partial_{xx}^2 + V^\chi_t - E^*) ( \partial_t u_+) + (\partial_t V^\chi_t) u_+ = 0.
    \]
    We multiply by $u_+$ and integrate between $x_0$ and $+\infty$ to get (recall that $u_+$ is exponentially decaying at $+ \infty$)
    \[
         \int_{x_0}^\infty (\partial_t V^\chi_t) | u_+ | ^2 = - \int_{x_0}^{\infty} u_+ ( - \partial_{xx}^2 + V^\chi_t - E^*) ( \partial_t u_+)
         = u_+' ( \partial_t u_+) - u_+ (\partial_t u_+') .
    \]
    For the last equality, we integrated by part and used again that $( - \partial_{xx}^2 + V^\chi_t - E^*) u_+ = 0$. Finally, we evaluate at $t = t^*$, and use that $(\partial_t V^\chi_t)(x) = 0$ for $ x < x_0 < -L$ to get
    \begin{equation*}
           E'_k(t^*) = \int_\R  (\partial_t V^\chi_t) | u_+ | ^2 = \int_{x_0}^\infty (\partial_t V^\chi_t) | u_+ | ^2 = u_+' ( \partial_t u_+) - u_+ (\partial_t u_+') \Big |_{t = t^*},
    \end{equation*}
    which is the numerator in~\eqref{eq:fraction}, up to a sign. This proves~\eqref{eq:IPPequality}.
\end{proof}

\subsection{The Dirichlet spectral flow}
\label{ssec:spectralFlowDirichlet}

We finally consider the self-adjoint Dirichlet operator
\[
    H^\sharp_D(t) := - \partial_{xx}^2 + V(x - t), \quad \text{acting on $L^2(\R^+)$, with domain $H^2_0(\R^+)$}.
\]
As in the previous section, the essential spectrum of $H^\sharp_D(t)$ is independent of $t$, equals $\sigma_{\ess} ( H^\sharp_D ) = \sigma(H_0)$, and eigenvalues may appear in the essential gaps. We assume that the $n$-th gap is open, and we write again $\left( E_k(t) \right)_{k \in K}$ with $K \subset \Z$ these branches of eigenvalues. We denote by $\cI^\sharp_{D,n}$ the spectral flow of $H^\sharp_D(\cdot)$ in the $n$-th gap.

\begin{proposition} \label{prop:SDn}
    Assume the $n$-th essential gap is open. Then $\cS^\sharp_{D, n} = n$.
\end{proposition}

\begin{proof}
Let $E \in g_n$ be any fixed energy in the $n$-th gap, and let $u_0 \in \cL_{t=0}^+(E)$, where $\cL_t^+(E)$ is the (bulk) vectorial space of solutions as introduced in Lemma~\ref{lem:LEpm}. Then $u_t(x) := u_0(x - t)$ is in $\cL_t^+(E)$ for all $t \in [0,1]$, solves the bulk ODE $-u_t'' + V_t u_t = E u_t$ and is exponentially decaying at $+ \infty$. Since $u_t$ is integrable at $+ \infty$,  it is an eigenvalue of $H^\sharp_D$ iff $u_t(0)= 0$, that is $u_0( - t) = 0$. We infer that $E$ is an eigenvalue of $H^\sharp_D(t)$ iff $-t \in \cZ_0^+$. So the number of edge modes at $E$ equals the number of $0$ of $u_0$ in the interval $[-1, 0]$, and we already proved in the proof of Proposition~\ref{prop:bulk=n} that this number equals $n$.

\medskip

It remains to prove that if $t \in E_k^{-1}(\{ E \})$, then $E_k'(t) \neq 0$. Let $t \in E_k^{-1}(\{ E \})$. We repeat the steps of the proof of Proposition~\ref{prop:Schi=n}. The function $u_t$ satisfies the ODE $( - \partial_{xx}^2 + V_t - E) u_t = 0$. Differentiating with respect to $t$ gives
\[
     ( - \partial_{xx}^2 + V_t - E) (\partial u_t) + (\partial_t V_t) u_t = 0.
\]
We multiply by $u_t$ and integrate over $\R^+$ to get
\[
   \int_{\R^+} u_t ( - \partial_{xx}^2 + V_t - E) (\partial_t u_t)  = - \int_{\R^+} (\partial_t V_t) | u_t |^2 = -E_k'(t),
\]
where we used again the Hellman-Feynman inequality for the last part (see Eqt.~\eqref{eq:HellmanFeynman}). Integrating by part the left-hand side gives
\[
    E_k'(t) = (\partial_x u_t) (\partial_t u_t)(x=0) = - | \partial_x u_t |^2(x = 0),
\]
where we used that $u_t(x) = u_0( x- t)$, so that $\partial_t u = - \partial_x u$. Finally, since $u_t(0)= 0$, we have $u_t'(0) \neq 0$, so $E_k'(t) < 0$.
\end{proof}

We actually proved that the branches of eigenvalue of $H^\sharp_D(\cdot)$ are decreasing functions of $t$.

\begin{remark}[Resonant states]
    One can perform the same analysis for the vectorial spaces $\cL^-_t(E)$. A solution $u_t \in \cL^-_t(E)$ satisfying the Dirichlet boundary condition $u_t(0) = 0$ is called a {\em resonant state}. We find that the spectral flow of resonant modes is $-n$ is the $n$-th gap, and that the corresponding curves are increasing. Actually, one can prove (see {\em e.g.}~\cite{korotyaev2000lattice}) that the combination of the two spectral flows gives a smooth curve (see Figure~\ref{fig:dirichlet_sf}).
\end{remark}

\begin{remark}[Spectral pollution]
    An interesting corollary of $\cS^\sharp_{D, n} = n$ is that if one numerically studies the periodic Hamiltonian $H(0)$ on a large box with Dirichlet boundary conditions, then spurious eigenvalues will appear. More specifically, on a box $[t, L+t]$ with $L$ large enough, there will be flows of {\em spurious} eigenvalues in all essential gaps, as $t$ goes from $0$ to $1$, corresponding to the localised edge modes near the boundaries $t$ and $L + t$. 
\end{remark}

\subsection{Numerical illustrations}

We end this section on the Schrödinger case with some numerical illustrations. We plot in Figure~\ref{fig:hamiltonian} the spectra of $t \mapsto H^\sharp_\chi(t)$  and $H^\sharp_D(t)$ in the special case 
\[
    V(x) := 50 \cos(2 \pi x) + 10 \cos(4 \pi x).
\]
We chose the potential $V$ so that the corresponding Hamiltonians have their first three essential gaps open and rather large. For the domain wall Hamiltonian, we took the simple continuous piece-wise linear cut-off function $\chi(x) := \1 \left( x \le -\frac12 \right) + \left( \frac12 - x \right) \1 \left(- \frac12 < x \le \frac12 \right)$. The bands are shown in grey. 

From Figure~\ref{fig:domainwallsf}, we plot the spectrum of the domain wall Hamiltonian $H^\sharp_{\chi}$. We see that the spectral flows is $n$ is the $n$-th gap. Some extra eigenvalues may appear in the gaps, but they do not contribute to the spectral flow. The shapes of the branches of eigenvalues depend on the cut-off $\chi$.

In Figure~\ref{fig:dirichlet_sf}, we plot the spectrum of the Dirichlet Hamiltonian $H^\sharp_D$. The solid blue lines show the decreasing spectral flow of the eigenvalues, while the dotted black lines show the increasing spectral flow of the resonant modes. The combination of the two gives a smooth curve.

\begin{figure}[h]
    \centering
    \begin{subfigure}{0.49\textwidth}
        \includegraphics[width=\textwidth]{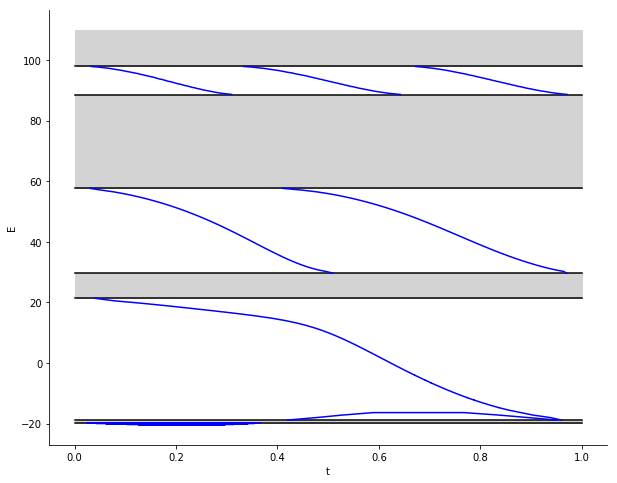}
        \caption{Domain wall Hamiltonian $t \mapsto H^\sharp_\chi(t)$.}
        \label{fig:domainwallsf}
    \end{subfigure}
    \begin{subfigure}{0.49\textwidth}
         \includegraphics[width = \textwidth]{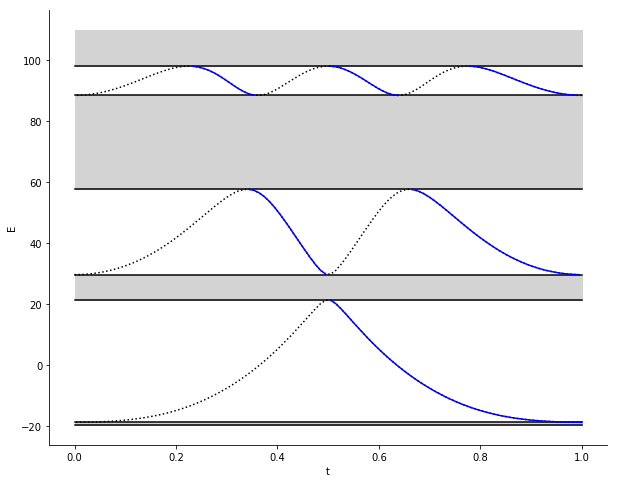}
        \caption{Dirichlet Hamiltonian $t \mapsto H^\sharp_D(t)$.}
        \label{fig:dirichlet_sf}
    \end{subfigure}
    \caption{Spectra of the domain wall (left) and Dirichlet (right) Hamiltonian as a function of $t$.  The essential spectrum is in grey, eigenvalues are in solid blue and resonant modes are in dotted black.}
    \label{fig:hamiltonian}
\end{figure}

\section{Dislocations in the Dirac case}
\label{sec:DiracCase}

We now focus on the Dirac case. We introduced the usual Pauli matrices $\bsigma_1, \bsigma_2$ and $\bsigma_3$, and the identity $\bone$, defined respectively by
\[
\bsigma_1 := \begin{pmatrix}
0 & 1 \\ 1 & 0
\end{pmatrix}, \quad 
\bsigma_2 := \begin{pmatrix}
0 & - \ri \\ \ri & 0
\end{pmatrix}, \quad 
\bsigma_3 := \begin{pmatrix}
1 & 0 \\ 0 & -1
\end{pmatrix}, 
\quad \text{and} \quad
\bone := \begin{pmatrix}
1 & 0 \\ 0 & 1
\end{pmatrix}.
\]
These matrices satisfies the relations $\bsigma_1^2 = \bsigma_2^2 = \bsigma^3 = \bone$, and
\[
\bsigma_1 \bsigma_2= - \bsigma_2 \bsigma_1 = - \ri \bsigma_3, \quad
\bsigma_2 \bsigma_3= - \bsigma_3 \bsigma_2 = - \ri \bsigma_1, \quad \text{and} \quad
\bsigma_3 \bsigma_1= - \bsigma_1 \bsigma_3 = - \ri \bsigma_2.
\] 
We consider $V \in L^1_\per(\R)$, and define the operator
\[
    \cD_0 := ( - \ri \partial_x) \bsigma_3 + V(x) \bsigma_1 \quad \text{acting on $L^2(\R, \C^2)$, with domain $H^1(\R, \C^2)$}.
\]
The proof of the following result is similar to the one of Lemmas~\ref{lem:Spectrum_Hill} and~\ref{lem:LEpm}. We provide a proof in Appendix~\ref{appendix:TransfertMatrix_Dirac} for completeness.
\begin{lemma} \label{lem:D0_selfadjoint}
    The operator $\cD_0$ is self-adjoint. Its spectrum is purely essential and symmetric with respect to the origin. For all $E \notin \sigma(\cD_0)$, the $\C$-vectorial space of solutions $\cL(E)$ is of dimension $2$, and has a splitting $\cL(E) = \cL^+(E) \oplus \cL^-(E)$, where $\cL^\pm(E)$ is the vectorial space of solutions that are square integrable at $\pm \infty$. In addition,
    \begin{itemize}
        \item The spaces $\cL^\pm(E)$ are of dimension $1$, and $\cL^+(E) \cap \cL^-(E) = \emptyset$;
        \item The solutions in $\cL^\pm(E)$ are exponentially decaying at $\pm \infty$;
        \item The maps $E \mapsto \cL^{\pm} (E)$ are differentiable in the sense of Definition~\ref{def:contL}.
    \end{itemize}
\end{lemma}

The symmetry of the spectrum comes from the fact that $\sigma_2 \cD_0 \sigma_2 = - \cD_0$, so $\cD_0$ is unitary equivalent to $-\cD_0$.

\begin{remark}~\label{rem:generalPotentials}
    Our results are valid with potentials $V$ which are not necessarily periodic. Actually, according to the following proofs, we only require the existence of a splitting $\cL(E) = \cL^+(E) \oplus \cL^-(E)$ for all $E \notin \sigma(\cD_0)$. 
\end{remark}

\subsection{Bulk index for the Dirac operator}
\label{ssec:Dirac_bulkIndex}
For $t \in \TT^1$, we define
\[
\cD(t) := ( - \ri \partial_x) \bsigma_3 + \re^{ - \ri t \pi \bsigma_3} \left( V(x) \bsigma_1 \right) \re^{ \ri t \pi \bsigma_3} \quad \text{acting on $L^2(\R, \C^2)$, with domain $H^1(\R, \C^2)$},
\]
where $\re^{ \ri t \pi \bsigma_3}$ is the unitary $2 \times 2$ matrix
\begin{equation} \label{eq:def:Rt}
     \re^{ \ri t \pi \bsigma_3}  = \begin{pmatrix}
        \re^{ \ri \pi t} & 0 \\ 0 & \re^{ -\ri \pi t}
    \end{pmatrix} = 
    \cos(\pi t) \bone + \ri \sin(\pi t) \bsigma_3.
\end{equation}
Since $\re^{ \ri t \pi \bsigma_3}$ commutes with $\bsigma_3$, the operator $\cD(t)$ can also be written as
\begin{equation} \label{eq:exact_Dt}
    \cD(t) = \re^{ - \ri t \pi \bsigma_3} \cD_0 \re^{  \ri t \pi \bsigma_3} = ( - \ri \partial_x) \bsigma_3 + \cos(2 \pi t) V \bsigma_1 - \sin(2 \pi t) V \bsigma_2. 
\end{equation}
In particular, $\cD(t)$ is $1$-periodic in $t$, and is a unitary transform of $\cD_0 := \cD(t = 0)$, hence they share the same purely essential spectrum.

\medskip

In what follows, we assume that $\sigma(\cD_0) \neq \R$, we fix $g \subset \R$ an essential gap of $\cD(\cdot)$, and we let $E \in g$. We see the equation $\cD(t) \bu = E \bu$ as a first order ODE on $\C^2$. In the sequel, we write $\bu = (u^\up, u^\down) \in \C^2$. If $\bu$ is non-null solution of $(\cD(t)  -  E) \bu = 0$, then $u^\up$ and $u^\down$ cannot vanish at the same time. We would like to introduce as before the function $\theta$ in~\eqref{eq:def:theta}, but $u^\up \pm \ri u^\down$ may vanish this time, as $u^{\up/\down}$ are now complex-valued functions. 

\medskip

For $t \in \TT^1$, we denote by $\cL_t^\pm(E)$ the $1$-dimensional vectorial space of solutions that are square integrable at $\pm \infty$. It holds $ \cL_t^\pm(E) = \re^{- \ri \pi t \bsigma_3} \cL_0^\pm(E)$.

\begin{lemma} \label{lem:up=down}
    If $\bu \in \cL^\pm_t(E)$ is a non-null solution, then $| u^\down | = | u^\up |$. In particular, the two functions $u^\down$ and $u^\up$ never vanish.
\end{lemma}

\begin{proof}
    From~\eqref{eq:exact_Dt}, $\cD(t)$ commutes with $\bsigma_1 K$, where $K \bu := \overline{\bu}$ is the complex conjugation operator. So if $\bu$ is a solution in $\cL^\pm_t(E)$, then so is $\bsigma_1 \overline{\bu}$, and there is $\lambda \in \C \setminus \{ 0 \}$, so that $\bsigma_1 \overline{\bu} = \lambda \bu$, that is 
    \[
        \overline{u^\up} = \lambda u^\down \quad \text{and} \quad
        \overline{u^\down} = \lambda u^\up.
    \]
    This implies $| \lambda | = 1$, and the result follows.
\end{proof}

We can therefore introduce the quantity
\begin{equation} \label{eq:def:theta_Dirac}
    \theta^+_t(E, x) := \theta[\cL^+_t(E), x] := \theta[\bu_t, x] := \dfrac{u_t^\down}{u_t^\up} (x) \quad \in \SS^1.
\end{equation}
where $\bu_t$ is any non-null function in $\cL^+_t(E)$, and similarly for $\theta^-_t(E, x)$. By construction, $\theta^\pm_t(E, x)$ are well-defined functions with value in $\SS^1$, and are continuous in $x \in \R$ and $E \in g$. As in Section~\ref{ssec:MaslovIndex}, the maps $t \mapsto \cL^\pm_t(E)$ are $1$-periodic, and we can define the Maslov indices as the winding number of $t \mapsto \theta_t^\pm$:
\[
    \cM^\pm := W \left[ t \mapsto \theta^\pm_t(E, x) \right], 
\]
which is independent of $x \in \R$ and $E \in g$, by continuity of $\theta^\pm_t$. As in Lemma~\ref{lem:M+=M-}, we have $\cM^+ = \cM^-$, and we define by $\cB := \cM^+ = \cM^-$ the bulk index.
\begin{proposition} 
    In all open gaps of $\cD_0$, we have $\cB = 1$.
\end{proposition}
\begin{proof}
    Let $\bu_0 \in \cL_0^+(E)$. Then $\bu_t := \re^{ - \ri \pi t \bsigma_3} \bu_0$ is in $\cL_t^+(E)$. From~\eqref{eq:def:Rt}, we obtain
    \[
        \theta^+_t(x) = \dfrac{u_t^\down}{u_t^\up}(x) = \dfrac{\re^{ \ri \pi t} u_0^\down}{\re^{ -\ri \pi t} u_0^\up} = \re^{ 2 \ri \pi t} \theta^+_0(x),
    \]
    which has winding $1$, as wanted.
\end{proof}

\subsection{The edge index in the Dirac case}
\label{ssec:Dirac_edgeIndex}

We now consider the self-adjoint domain wall Dirac operator
 \[
    \cD^\sharp_\chi(t) = \cD_0 \chi + \cD(t) (1 - \chi)  
    = (- \ri \partial_x ) \bsigma_3 + V \bsigma^\chi_t
    \quad
\text{acting on $L^2(\R)$, with domain $H^2(\R)$},
\]
where we introduced
\[
    \bsigma^\chi_t(x) := \left[ \chi(x) + (1 - \chi(x)) \cos(2 \pi t) \right] \bsigma_1 - (1 - \chi(x)) \sin( 2 \pi t) \bsigma_2.
\]
The map $t \mapsto \cD^\sharp_\chi(t)$ is $1$-periodic in $t \in \TT^1$. It is classical that the essential spectrum of $\cD^\sharp_\chi(t)$ is $\sigma(\cD_0) \cup \sigma(\cD(t)) = \sigma(\cD_0)$, and some eigenvalues may appear in the essential gaps. At $t = 0$ however, we have $\cD^\sharp_\chi(t=0) = \cD_0$, which has only purely essential spectrum.

\medskip

For $E \in g$ in a gap, we denote by $\cL^{\sharp, \pm}_{\chi,t}(E)$  the edge vectorial space of solutions of $\cD^\sharp_\chi(t) \bu = E \bu$ that are square integrable at $\pm \infty$. As in Proposition~\ref{prop:edge=n}, these spaces are both of dimension $1$.  They may have a non-trivial intersection, and if $\bu \in \cL^{\sharp, +}_{\chi, t}(E) \cap \cL^{\sharp, -}_{\chi,t}(E)$ is non null, then $\bu$ is a square integrable non null solution to $\cD^\sharp_\chi(t) \bu = E \bu$, {\em i.e.} an edge mode.

\medskip

In this edge situation, Lemma~\ref{lem:up=down} still holds, and we can introduce the edge quantity
\[
    \Omega^\sharp_{\chi}(t, E, x) := \theta^{\sharp, +}_{\chi,t}(E, x) \left(  \theta^{\sharp, -}_{\chi,t}(E, x) \right)^{-1}, 
    \quad \text{where} \quad
    \theta^{\sharp, +}_{\chi,t}(E, x) := \theta \left[ \cL^{\sharp, +}_{\chi,t}(E), x  \right],
\]
with $\theta$ as defined in~\eqref{eq:def:theta_Dirac}. The maps $t \mapsto \cL^{\sharp, \pm}_\chi(E)$ are $1$-periodic, so $\Omega^\sharp_\chi(\cdot, E, x)$ is a continuous map from $\TT^1$ to $\SS^1$, and we can define the edge index $\cI^\sharp_\chi$ as its winding number (see Definition~\ref{def:edgeIndex}). This index is independent of $E \in g$ and $x \in \R$ by continuity.

\begin{proposition} In all open gaps of $\cD_0$, we have $\cI^{\sharp}_\chi = 1$. In particular, it is independent of $\chi$.
\end{proposition}
\begin{proof}
    As in the proof of Proposition~\ref{prop:edge=n}, we have $\cM^{\sharp, +}_\chi = \cM^+ = 1$, while $\cM^{\sharp, -}_\chi = 0$.
\end{proof}
Results similar to Lemma~\ref{lem:spectrum_Hedge} and Lemma~\ref{lem:existenceEdgeModes} hold.  In particular, all eigenvalues of $\cD^\sharp_\chi(t)$ are simple.

\subsection{The Dirac spectral flow}
\label{ssec:Dirac_spectralFlow}

It remains to link the winding of $\Omega^\sharp_\chi[\cdot, E, x]$ to the spectral flow of $\left( \cD^\sharp_\chi(t) \right)_{t \in \TT^1}$, which we denote by $\cS^\sharp_\chi$. 

\begin{proposition}
    In all open gaps of $\cD_0$, We have $\cS^\sharp_\chi = \cI^\sharp_\chi = 1 $. In particular, it is independent of $\chi$.
\end{proposition}

\begin{proof}
    We denote again by $\left( E_k(t) \right)_{k \in K}$ the branches of eigenvalues of $\cD^\sharp_\chi(t)$, where $K$ is a subset of $\Z$. Since all eigenvalues are simple, the branches $E_k(\cdot)$ cannot cross. Let $E^* \in g_0$ be a regular point for the spectrum of $\cD^\sharp_\chi$, and let $t^*$ be such that $\Omega^{\sharp}_\chi[t^*, E^*, x_0 ] = 1$ for some fixed $x_0 < -L$, so that there is $k \in K$ with $E_k(t^*) = E^*$. As in the proof of Proposition~\ref{prop:Schi=n}, it is enough to prove that
    \[
        \nu_{z = 1} \left[\Omega^{\sharp}_\chi(\cdot, E^*, x_0 ), t^*  \right] = - {\rm sgn} \, E_k'(t^*).
    \]
    We start with $E_k'(t^*)$, which is non null. Let $\bu_t$ be an eigenvalue of $\cD^\sharp_\chi(t)$ for $E_k(t)$. With the Hellmann-Feynman theorem, we have
    \begin{equation} \label{eq:HellmannFeynmann_Dirac}
         E_k' (t)  = \partial_t \left\bra \bu_t, \cD^\sharp_\chi(t) \bu_t \right\ket 
        = \left\bra \bu_t, \partial_t \left[ \cD^\sharp_\chi(t) \right] \bu_t \right\ket
        = \int_{\R} V  \left[ \bu_t^*  (\partial_t \bsigma^\chi_t) \bu_t \right].
    \end{equation}
    On the other hand, from the definition of $\Omega^\sharp_\chi$ and the fact that $x_0 < -L$, we have
    \[
        - \ri \left( \partial_t \Omega^\sharp_\chi \right) (t^*, E^*, x_0) = 
         - \ri \dfrac{ \left( \partial_t \theta^{\sharp, +}_{\chi,t}\right) (E^*,x_0) }{\theta^{\sharp, +}_{\chi,t} (E^*,x_0) } \Big|_{t = t^*}
         = - \ri \left( \dfrac{\partial_t u_t^\down}{u^\down_t} -  \dfrac{\partial_t u_t^\up}{u^\up_t}  \right) (x_0).
    \]
    Using Lemma~\ref{lem:up=down}, this is also
    \[
        - \ri \left( \partial_t \Omega^\sharp_\chi \right) (t^*, E^*, x_0) 
        = - \ri \dfrac{\overline{u_t^\down} (\partial_t u_t^\down) - \overline{u_t^\up} (\partial_t u_t^\up)}{| u_t^\up |^2}(x_0) = \ri \dfrac{\bu_t^* \bsigma_3 (\partial_t \bu_t)}{| u_t^\up |^2}(x_0).
    \]
    Finally, to link these two quantities, we differentiate the ODE $\cD^\sharp_\chi(t) \bu_t = E^* \bu_t$ with respect to $t$, and get
    \[
        \left( (- \ri \partial_x) \bsigma_3 + V \bsigma^\chi_t - E^* \right) (\partial_t \bu_t ) + V (\partial_t \bsigma^\chi_t) \bu_t = 0.
    \]
    Multiplying by $\bu_t^*$ (which is exponentially decreasing at $+ \infty$) and integrating between $x_0$ and $+ \infty$ gives
    \[
        \int_{\R} V \left[ \bu_t^* (\partial_t \bsigma^\chi_t) \bu_t \right] = \int_{x_0}^\infty V \left[ \bu_t^* (\partial_t \bsigma^\chi_t) \bu_t \right]
        = - \int_{x_0}^\infty \bu_t^*  \left( (- \ri \partial_x) \bsigma_3 + V \bsigma^\chi_t - E^* \right) (\partial_t \bu_t ),
    \]
    where the first equality comes from the fact that $ (\partial_t \bsigma^\chi_t) (x) = 0$ for $x < x_0 < -L$. Finally, we integrate by part the last integral and use again that $(\cD^\sharp_\chi(t) - E^*) \bu_t = 0$ to get, as wanted
    \[
         E_k'(t) = \int_{\R} V \left[ \bu_t^* (\partial_t \bsigma^\chi_t) \bu_t \right] = - \ri \bu_t^* \bsigma_3(\partial_t \bu_t),
    \]
    and the proof follows.
\end{proof}

\subsection{The spectrum of $\cD^\sharp_\chi(t)$ at $t = \frac12$.} 
\label{ssec:Dirac_t=12}
We now focus on the spectrum of $\cD^\sharp_\chi(t = \tfrac12)$. First, we have the symmetry $\bsigma_2 \re^{ \ri t \pi \bsigma_3} \bsigma_2 = \re^{ - \ri t \pi \bsigma_3}$, so
\begin{equation*}
\bsigma_2 \cD_0 \bsigma_2 = -\cD_0,
\quad
\bsigma_2 \cD(t) \bsigma_2 = -\cD(-t),
\quad \text{and} \quad
\bsigma_2 \cD^\sharp_\chi(t) \bsigma_2 = -\cD^\sharp_\chi(-t).
\end{equation*}
The first relation shows that the spectrum of $\cD(0)$ (hence of $\cD(t)$ as well) is symmetric with respect to the origin. So if $g$ is a gap of $\cD_0$, then so is $-g$. The last relation shows that
\[
\sigma \left( \cD^\sharp_\chi(t)  \right) = - \sigma \left(\cD^\sharp_\chi(-t) \right).
\]
This means that the graph $\left(t, \sigma \left( \cD^\sharp_\chi(t)  \right) \right)_{t \in \TT^1}$ is symmetric with respect to the point $(0, 0)$ and to the point $(\tfrac12, 0)$ (see Figure~\ref{fig:diracdwsf} below). 

\medskip

Of particular importance is the middle gap. Let us assume that $0 \notin \cD_0$ is not in the spectrum of $\cD_0$, and let $g_0 \subset \R$ be the gap of $\cD_0$ containing $0$. This holds in particular for the constant case $V = \kappa \neq 0$, since $\cD_0^2 = ( - \partial_{xx}^2 + \kappa^2) \bone$, whose spectrum is $[\kappa^2, \infty)$. In this case, the spectrum of $\cD_0$ is $\sigma(\cD_0) = (-\infty, -\kappa] \cup [\kappa, \infty)$, and $g_0 = (-\kappa, \kappa)$.

\begin{proposition}
    If $0$ is not in the spectrum of $\cD_0$, then $0$ is a simple eigenvalue of $\cD^\sharp_\chi( t = \tfrac12)$.
\end{proposition}

\begin{proof}
    This is a direct consequence of the fact that the spectral flow of $\cD^\sharp_\chi$ is $1$ in $g_0$, combined with the fact that the graph of the spectrum $(t, \sigma(\cD^\sharp_\chi(t)))$ is symmetric with respect to $(0, \tfrac12)$. 
\end{proof}

\subsection{Numerical illustrations in the Dirac case}

We end this section with a numerical illustration in the Dirac case. In Figure~\ref{fig:diracdwsf}, we plot the spectrum of $\cD^\sharp_\chi(t)$ as a function of $t$. We took the potential
\[
    V(x) := 1 + \cos( 2 \pi x),
\]
for which $0$ is not in the essential spectrum. Actually, the Dirac operator $\cD_0$ has three gaps with this choice. For the cut-off function $\chi$, we took the continuous piece-wise linear function $\chi(x) := \1 \left( x \le -\frac12 \right) + \left( \frac12 - x \right) \1 \left(- \frac12 < x \le \frac12 \right)$. We see from this figure that the spectral flow is $1$ in each gap, and that $0$ is an eigenvalue for $t = \tfrac12$, as proved.

\begin{figure}[h]
    \centering
    \includegraphics[width=0.8\textwidth]{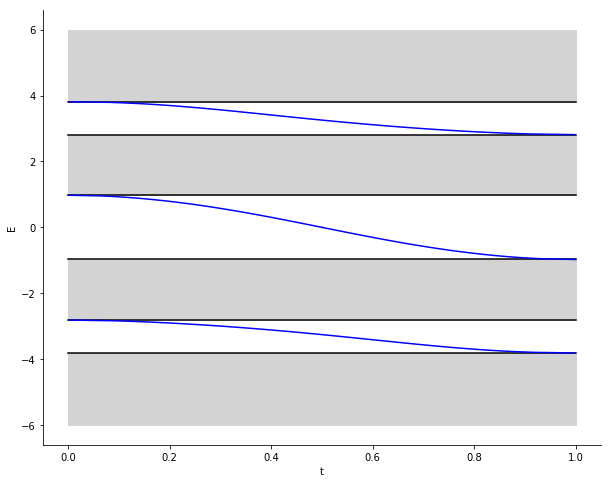}
    \caption{Spectrum of $t \mapsto \cD^\sharp_\chi(t)$.}
    \label{fig:diracdwsf}
\end{figure}

\section{Proofs concerning regularity}
\label{sec:proofs}

We gather in this section all the proofs concerning regularities.

\subsection{Proof of Lemma~\ref{lem:BasicFact_ODE}}
\label{ssec:proof_basicFact_ODE}
 Let us prove that $c_V$ is well-defined on $\R$. The proof is similar of $s_V$. It is enough to show that $c_V$ and $c_V'$ are bounded on $\R$. We set $y(x) := (c_V(x), c_V'(x))^T$, so that
    \[
    y' = A(x) y, \quad \text{with} \quad 
    A(x) := \begin{pmatrix}
    0 & 1 \\
    -V(x) & 0
    \end{pmatrix}.
    \]
    This gives
    \[
    \| y(x) \| \le \| y(0) \| + \| y(x) - y(0) \| \le \| y(0) \| + \left\| \int_0^x A(s) y(s) \rd s \right\| 
    \le \| y(0) \| + \int_0^x \left\| A(s) \right\| \cdot \| y(s) \| \rd s.
    \]
    Since $\| A(\cdot) \|$ is locally integrable together with the Grönwall lemma, we conclude that $y$ stays bounded on $\R$.
    The rest of the proof is straightforward. The last assertion is a consequence of Cauchy-Lipschitz.

\subsection{Proof of Lemma~\ref{lem:ct_and_st_are_periodic}}
\label{ssec:proof:ct_and_st_are_periodic}
    Again, we focus on $c_t$. Periodicity in $t$ is straightforward, and we already proved in Lemma~\ref{lem:BasicFact_ODE} that for all $t \in \TT^1$, the map $c_t(\cdot)$ is continuously differentiable. It remains to prove that for all $x \in \R$, $t \mapsto c_t(x)$ is also continuously differentiable. Differentiating the equation $(-c_t'' + V_t c_t ) = 0$ with respect to $t$ gives
    \[
        \left( - \partial_{xx}^2 + V_t(x) \right) (\partial_t c_t)(x) + \left( \partial_t V_t \right) c_t = 0.
    \]
    Hence the function $\partial_t c_t$ solves a second order ODE with second member. Also, we have $\partial_t c_t(0) = 0$ and $\partial_t c_t'(x)= 0$. The solution is explicitly given by
    \begin{align*}
    \partial_t c_t (x) = \int_0^x \left[ c_{t} (y) s_{t}(x) - c_{t} (x) s_{t}(y)  \right]  \left( \partial_t V_t \right)(y) c_{t}(y) \rd y.
    \end{align*}
    Since $(\partial_t V_t)$ is in $L^1_\loc$, while $c_t$ and $s_t$ are continuous, the left-hand side is indeed continuous, which concludes the proof.

\subsection{Proof of Lemma~\ref{lem:continuityOfL}}
\label{ssec:proof:lem_continuityOfL}
    If $c_0$ and $s_0$ are the fundamental solutions of $(-\partial_{xx}^2 + V_{t = 0}) u = Eu$, the $c_t := c_0(\cdot - t)$ and $s_t := s_0 \cdot( - t)$ are the ones of $(-\partial_{xx}^2 + V_{t}) u = Eu$. So if $\cL_{t = 0}^\pm(E) = \lambda_c^\pm c_0 + \lambda_s^\pm s_0$, then $\cL_t^\pm(E) = \lambda_c^\pm c_t + \lambda_s^\pm s_t$. Identifying with Definition~\ref{def:contL}, we see that the maps $\lambda_c(\cdot)$ and $\lambda_s(\cdot)$ are constant, so $t \mapsto \cL_t^\pm(E)$ are continuously differentiable.

\subsection{Proof of Lemma~\ref{lem:regularityPn}}
\label{ssec:proof:regularityPn}

Let us prove that $\| \partial_t P \|_{\cB(L^2([0, 1]))} < \infty$ (the proof for $\partial_k P$ is similar and standard). Let $\sC$ be positively oriented contour in the complex plane enclosing the $n$ first bands of $H_0$. From Cauchy's residual formula, we have
\[
    P(t,k) = \dfrac{1}{2 \ri \pi} \oint_\sC R(t, z) \rd z, \quad \text{with} \quad R(t, z) := \dfrac{1}{z - H(t,k)}.
\]
Differentiating with respect to $t$ gives
\[
    (\partial_t P)(t,k) = \dfrac{1}{2 \ri \pi} \oint_\sC R(t, z) \left( \partial_t V_t \right)(t) R(t,z) \rd z.
\]
For all $(t, z)$ in the compact $\TT^1 \times \sC$, the operator $R(t, z)$ is bounded from $L^2([0,1]) \to H^2_k$. Also, we have $L^\infty([0, 1]) \hookrightarrow H^2_k$. Hence there is $C \in \R^+$ so that
\[
   \forall u \in L^2([0,1]), \  \forall (t, z) \in \TT^1 \times \sC, \quad \left\| R(t, z) u \right\|_{L^\infty([0, 1])} \le C \| u \|_{L^2([0,1])}.
\]
This gives, for all $u, v \in L^2([0, 1])$.
\begin{align*}
    \left| \bra u, (\partial_t P), v \ket_{L^2([0, 1])} \right| & = 
    \left| \dfrac{1}{2 \ri \pi} \oint_\sC \bra R(t, z)^* u, \left( \partial_t V_t \right) , R(t,z) v \ket \rd z \right| \\
    & \le \dfrac{1}{2 \pi } \oint_{\sC} \left\| \partial_t V_t \right\|_{L^1_\per} \left\| R(t, z)^* u \right\|_{L^\infty([0, 1])} \left\| R(t, z) v \right\|_{L^\infty([0, 1])}   \rd z \\
    & \le \dfrac{C^2 | \sC |}{2 \pi}  \left\| \partial_t V_t \right\|_{L^1_\per} \| u \|_{L^2([0, 1])} \| v \|_{L^2([0, 1])},
\end{align*}
which proves that $\| \partial_t P \| \le \frac{C^2 | \sC |}{2 \pi} \left\| \partial_t V_t \right\|_{L^1_\per} < \infty$, as wanted.

\subsection{Proof of Lemma~\ref{lem:continuityOfEigenvalues}}
\label{ssec:proof:continuityOfEigenvalues}
This is a standard argument in perturbation theory for operators (see {\em e.g.}~\cite{kato2013perturbation}). For $t^* \in \TT^1$. The eigenvalues of $H_\chi^\sharp(t^*)$ can only accumulate near the band edges $E_n^-$ and $E_{n+1}^+$. Hence, all the eigenvalues are well separated (we recall that they are simple by Lemma~\ref{lem:spectrum_Hedge}). Let $\sC$ be a positively oriented loop in the complex plane enclosing the eigenvalue $E^*$, and none other. The projector on the corresponding eigenvector is $P(t^*)$, with
    \[
    P(t^*) := \dfrac{1}{2 \ri \pi} \oint_{\sC} \dfrac{\rd z}{z - H_\chi^\sharp(t^*)}.
    \]
    This formula is well-defined and continuous in a neighbourhood of $t^*$. Since $\Ran \, P(t) = \Tr P(t)$ is continuous and integer-valued, we have $\Ran P(t) = 1$ in this neighbourhood, and we infer that there is a unique eigenvalue $E(t)$ of $H^\sharp_\chi(t)$ in the contour $\sC$. As $H^\sharp_\chi(t)$ is self-adjoint, $E(t)$ is real-valued. Finally, from the formula
    \[
        E(t) = \Tr( H^\sharp_\chi(t) P(t)),
    \]
    we see that the map $t \mapsto E(t)$ is continuous. This argument can be repeated in a maximal interval $0 \le t^- < t^* < t^+ \le 1$, as long as $E(t)$ does not touch the band gaps. Also, since at $t = 0$, $H^\sharp_\chi(t = 0)= H_0$ has no eigenvalue, we have $0 < t^- < t^+ < 1$ and the result follows.


\appendix

\section{The spectrum of periodic operators}
\label{appendix:periodic_operator}

In this Appendix, we recall for completeness some basic facts about periodic ODEs.

\subsection{Transfer matrix for Hill's operators.} 
\label{appendix:TransfertMatrix}

We fix $V \in L^1_\per(\R)$ and focus on the solutions  $-u'' + V(x) u = Eu$. Some references for the properties of such ODEs and related Schrödinger (or Hill's) operators are~\cite{reed1978analysis, poschel1987}. We denote by $c_E$ and $s_E$ the fundamental solutions introduced in~\eqref{eq:def:fundamental_solutions}. The {\em transfer matrix} is the $2 \times 2$ matrix defined by $T_E := T_E(x = 1)$, where
\begin{equation*} 
T_E(x)  := \begin{pmatrix}
c_E(x) & s_E(x) \\
c_E'(x) & s_E'(x)
\end{pmatrix}.
\end{equation*}
Its {\em discriminant} is $\Delta(E) := \Tr(T_E)$. If $u$ satisfies $-u'' + Vu = Eu$, then 
\begin{equation} \label{eq:T^n}
\forall x \in \R, \quad
u(x) = u(0) c_E(x) + u'(0) s_E(0), \quad \text{so that} \quad
\begin{pmatrix}
u(x) \\u'(x)
\end{pmatrix}
=
T_E(x) \begin{pmatrix}
u(0) \\
u'(0)
\end{pmatrix}.
\end{equation}
For $n \in \Z$, we denote by $\tau_n f := f(\cdot - n)$ the translation operator. Since $V$ is periodic, if $u$ is a solution, then so is $\tau_n u$. This implies that $T_E(x+n) = T_E^n \cdot T_E(x)$, hence the asymptotic behaviour of the solutions are determined by the singular values of $T_E$. 
\begin{lemma} \label{lem:transfert_matrix}
    For all $x \in \R$, we have $\det T_E(x) = 1$. In particular, if $\lambda \in \C^*$ is an singular value of $T_E$, then so is $\lambda^{-1}$. Furthermore,
    \begin{enumerate}[(i)]
        \item \label{case:tr<2} If $\left| \Delta(E) \right| > 2$, then $\lambda, \lambda^{-1} \in \R^*$. In this case, any non null solution is exponentially decaying either at $+ \infty$ or at $-\infty$ (and exponentially increasing in the other direction);
        \item \label{case:tr>2} If $\left| \Delta(E) \right| \le 2$, then $\lambda \in \SS^1$, and $\lambda^{-1} = \overline{\lambda}$. In this case, all solutions are bounded on $\R$.
    \end{enumerate}
\end{lemma}

\begin{proof}
    The determinant $\det T_E(x) = c_E s_E' - c_E' s_E$ is the Wronskian of $c_E$ and $s_E$. Differentiating gives
    \[
    \left( \det T_E \right)' =  c_E s_E'' - c_E'' s_E = c_E \left( V - E \right) s_E - (V - E) c_E s_E = 0.
    \]
    Hence $\det T_E(x)$ is constant, and equals $\det T(x = 0) = 1$. The rest of the proof of Lemma~\ref{lem:transfert_matrix} follows from the identity $\Delta(E) = \Tr (T_E) = \lambda + \frac{1}{\lambda} = c_E + s_E' \in \R$.
\end{proof}
The map $E \mapsto V - E$ is continuously differentiable in $L^1_\loc(\R)$. Lemma~\ref{lem:ct_and_st_are_periodic} then shows that $E \mapsto T_E$ is continuous on $\R$ (it is actually analytic, see {\em e.g.}~\cite[p.10]{poschel1987}).

\medskip

Closely related to the ODE $-u'' + Vu = Eu$ is the Hill's operator
\begin{equation*}
H_0 := - \partial_{xx}^2 + V \quad \text{acting on $L^2(\R)$, with domain $H^2(\R)$.} 
\end{equation*}
The operator $H$ is self-adjoint, and commutes with the translations $\tau_n$. Its properties can be studied from its Bloch transform. For $k \in \R$, we denote by $H(k)$ the Bloch fibers
\[
H(k) := - \partial_{xx}^2 + V \quad \text{acting on $L^2([0,1])$, with domain $H^2_k$},
\]
where the Hilbert spaces $H^2_k$ were defined in~\eqref{eq:def:H2k}. The spaces $H^2_k$ are $1$-periodic in $k$, hence so are the operators $H(k)$. For all $k \in \TT^1$, the operator $H(k)$ is compact resolvent and bounded from below. We denote by 
\[
\varepsilon_{1,k} \le \varepsilon_{2,k} \le \cdots \le \varepsilon_{n,k} \le \cdots
\]
its eigenvalues, counting multiplicity, and by $u_{n,k} \in H^2_k$ a corresponding basis of eigenfunctions. From the Bloch decomposition, we have
\[
\sigma(H_0)= \bigcup_{k \in \TT^1} \sigma \left( H_k \right) = \bigcup_{k \in \TT^1} \bigcup_{n \in \N} \left\{ \varepsilon_{n, k} \right\}.
\]

Seen as a function over $\R$, $u_{n,k}$ is solution to
\[
- u_{n,k}'' + V u_{n,k} = \varepsilon_{n,k} u_{n,k}, 
\quad \text{with} \quad  
u_{n,k}(x+1) = \re^{ 2 \ri \pi k} u_{n,k}(x)
\quad \text{and} \quad
u_{n,k}'(x+1) = \re^{ 2 \ri \pi k} u_{n,k}'(x).
\]
Together with~\eqref{eq:T^n}, we deduce that $\re^{2 \ri \pi k}$ is an eigenvalue of $T_{\varepsilon_{n,k}}$, with corresponding eigenvector $(u_{n,k}(x), u_{n,k}'(x))^T$ for any $x \in \R$. As a result:
\begin{lemma} \label{lem:spectrumWithTE}
    $ E \in \sigma \left( H(k) \right)$ iff $\Delta(E) = 2 \cos(k)$. In particular, $\sigma(H) = \Delta^{-1} ([-2, 2])$.
\end{lemma}
Since $E \mapsto T_E$ is continuous, then so is $E \mapsto \Delta(E)$. This implies that the spectrum of $H$ is composed of bands and gaps. More specifically, we have the following (see~\cite{pankrashkin2014remark} for a simple proof, or~\cite[Theorem XIII.89]{reed1978analysis}).
\begin{lemma} 
    The functions $k \mapsto \varepsilon_{n,k}$ are continuous on $\R$, satisfy $\varepsilon_{n,-k} = \varepsilon_{n,k}$ and $\varepsilon_{n, k + 2 \pi} = \varepsilon_{n, k}$. In addition, there is a sequence of intervals $[E_n^-, E_n^+]$ for $n \in \N^*$ with $E_n^- < E_n^+$ and $E_n^+ \le E_{n+1}^-$ such that
    \begin{itemize}
        \item If $n$ is odd, then $k \mapsto \varepsilon_{n,k}$ is increasing on $[0, \pi]$, with $\varepsilon_{n,0}= E_n^-$ and $\varepsilon_{n, \pi}= E_n^+$;
        \item If $n$ is even, then $k \mapsto \varepsilon_{n,k}$ is decreasing on $[0, \pi]$, with $\varepsilon_{n,0}= E_n^+$ and $\varepsilon_{n, \pi}= E_n^-$.
    \end{itemize}
    In particular, $\sigma \left( H \right) = \bigcup_{n \in \N^*} \left[ E_n^-, E_n^+ \right]$, and the sequence $\{ E_1^- < E_1^+ \le E_2^-  < E_2^+ \le \cdots \}$ are the eigenvalues of $- \partial_{xx}^2 + V$ on $L^2([0, 2])$ with periodic boundary conditions.
\end{lemma}

\subsection{Transfer Matrix for Dirac's operators}
\label{appendix:TransfertMatrix_Dirac}
For the Dirac equation $( - \ri \partial_x) \bsigma_3 \bu + V(x) \bsigma_1 \bu = E \bu$, one can perform a similar analysis. We denote by $\bc_E$ and $\bs_E$ the solution of the Cauchy problem with initial values $\bc_E = (1, 0)^T$ and $\bs_E = (0, 1)^T$, and we define the {\em transfer matrix} $T_E := T_E(1)$ where $T_E(x) = (\bc_E(x), \bs_E(x) ) \in \cM_{2 \times 2}(\C)$. Its {\em discriminant} is $\Delta(E) := \Tr(T_E)$. If $\bu$ satisfies $( - \ri \partial_x) \bsigma_3 \bu + V(x) \bsigma_1 \bu = E \bu$, then 
\[
    \forall x \in \R, \quad \bu(x) = u_1 \bc_E(x) + u_2 \bs_E(x) = T_E \bu(0).
\]
Since $V$ is $1$-periodic, we have $T_E(x + n) = T_E^n T_E(x)$, so the behaviour of the solutions at infinity depends on the singular values of $T_E$.

\begin{lemma}
    We have $\det (T_E ) = 1$, and $\Delta(E) \in \R$.
\end{lemma}

\begin{proof}
    We have $ \det(T_E)(x) = c_{1,E}(x) s_{2, E}(x) - c_{2, E}(x) s_{1, E}(x) = \ri \bc_E^T \bsigma_2 \bs_E$. Differentiating gives
    \[
        \det(T_E)'(x) = \ri (\bc_E')^T \bsigma_2 \bs_E + \ri \bc_E^T \bsigma_2 \bs_E'.
    \]
    If $( - \ri \partial_x) \bsigma_3 \bu + V \bsigma_1 \bu = E \bu$, then $\bu' + V  \bsigma_2 \bu = \ri E \bsigma_3 \bu$. This gives
    \begin{align*}
        \det(T_E)'(x) & = \ri \left[ (\ri E \bsigma_3 - V \bsigma_2) \bc_E \right]^T \bsigma_2 \bs_E + \ri \bc_E^T \bsigma_2 (\ri E \bsigma_3 - V \bsigma_2) \bs_E \\
        & = - E \bc_E^T \bsigma_3 \bsigma_2 \bs_E  + \ri \bc_E^T V \bsigma_2 \bsigma_2 \bs_E
             - E \bc_E^T \bsigma_2 \bsigma_3 \bs_E - \ri \bc_E^T V \bsigma_2 \bsigma_2 \bsigma_2 \bs_E = 0.
    \end{align*}
    The determinant is therefore constant, and equals its value at $x = 0$, that is $\det T(x) = 1$. 
    
    To prove that the trace of $T_E(x)$ is real, we remark that if $\bu$ is a solution to $( - \ri \partial_x) \bsigma_3 \bu + V \bsigma_1 \bu = E \bu$, then $\bsigma_1 \overline{\bu}$ is also a solution. We deduce that $\bsigma_1 \overline{\bc_E} = \bs_E$, and finally that $\Tr(T_E) = c_{1,E} + s_{2, E} = c_{1, E}  + \overline{c_{1,E}} = 2 \Re \, c_{1, E} \in \R$.
\end{proof}
We can now repeat the arguments to have results similar to Lemma~\ref{lem:transfert_matrix} and Lemma~\ref{lem:spectrumWithTE}.

\section{Hill's operators on a segment}
\label{sec:appendix:Spectra}

In this appendix, we study the spectrum of
\[
    H^D := - \partial_{xx}^2 + V \quad \text{acting on $L^2([0, 1]$, with domain $H^2_0([0, 1])$},
\]
that is with Dirichlet boundary conditions\footnote{Our analysis can be repeated {\em mutatis mutandis} to study the operator with Neumann, Robin, or periodic boundary conditions.}. This operator is compact resolvent and bounded from below, hence have discrete spectrum. We denote by $\delta_1 \le \delta_2 \le \cdots$ its eigenvalues, ranked in increasing order, and by $(f_n)_{n \in \N^*}$ a respective basis of eigenvectors. The next lemma is very similar to~\cite[Thm 6]{poschel1987}.

 \begin{lemma}
     The eigenvalues $\delta_n$ are all simple. In addition, For all $n \in \N^*$, the function $f_n$ vanishes $n$ times  on $[0, 1)$.
\end{lemma}
\begin{proof}
    By contradiction, if $\delta_n$ is an eigenvalue with double multiplicity, and $y_1$, $y_2$ two corresponding eigenvectors, then $y_1$ and $y_2$ are linearly independent, and they solve $- y_{i}'' + (V-E)y_i = 0$. Hence, any solution to $-u'' + (V-E)u = 0$ on $\R$ is a linear combination of $y_1$ and $y_2$, and in particular vanishes at $x = 0$. This contradicts the Cauchy-Lipschitz theorem.
    
    \medskip
    
    We now prove the second part. By usual perturbation theory~\cite{kato2013perturbation}, there is a continuous map $s \mapsto \delta_n(s)$ such that $\delta_n(s)$ is the $n$-th eigenvalue of $H^D(s) := - \partial_{xx}^2 + sV$ (we recall that the eigenvalues are simple, hence cannot cross). We denote by $f_{n,s}$ the corresponding eigenvector, normalised so that $f_{n,s}'(0) = 1$, and by $\cZ(s)$ the set of zeros of $f_{n,s}$ in the interval $[0, 1)$. The set $\cZ(s)$ is finite, continuous, and since two zeros can never merge nor vanish, its cardinal is independent of $s \in [0, 1]$. So ${\rm Card}\,  \cZ(1) = {\rm Card} \, \cZ(0)$. When $s = 0$, we recover the usual one-dimensional Dirichlet Laplacian, hence $f_{n,0} = \frac{1}{n \pi} \sin(n \pi x)$, which vanishes $n$ times, and the result follows.
\end{proof}

\begin{lemma}
    The $n$-th eigenvalue of $H^D$ lies in the $n$-th gap of $H_0$: $E_n^+ \le \delta_n \le E_{n+1}^-$.
\end{lemma}

\begin{proof}
    At the energy $E = \delta_n$, we have $s_E(x = 1) = 0$ (the second fundamental solution is a Dirichlet solution). So $T_E$ is lower diagonal, and its singular values are the real quantities $c_E(1)$ and $s_E'(1)$. Hence $\Delta(E) \ge 2$, and we infer from Lemma~\ref{lem:spectrumWithTE} that $\delta_n \notin \sigma(H_0)$. We now consider the previous deformation, replacing $V$ by $sV$.  We proved that $\delta_n(s)$ cannot enter the band regions as $s$ goes from $0$ to $1$. Since the result is valid for $s = 0$, it remains true at $s = 1$, and the result follows.
\end{proof}


\begin{thebibliography}{ASBVB13}
    
    \bibitem[ASBVB13]{avila2013topological}
    {\sc J.~Avila, H.~Schulz-Baldes, and C.~Villegas-Blas}, {\em Topological
        invariants of edge states for periodic two-dimensional models}, Math. Phys.,
    Analysis and Geometry, 16 (2013), pp.~137--170.
    
    \bibitem[Bal17]{bal2017topological}
    {\sc G.~Bal}, {\em Topological protection of perturbed edge states}, arXiv
    preprint arXiv:1709.00605,  (2017).
    
    \bibitem[Bal18]{bal2018continuous}
    \leavevmode\vrule height 2pt depth -1.6pt width 23pt, {\em Continuous bulk and
        interface description of topological insulators}, arXiv preprint
    arXiv:1808.07908,  (2018).
    
    \bibitem[CGLM19]{cornean2019localised}
    {\sc H.~Cornean, D.~Gontier, A.~Levitt, and D.~Monaco}, {\em Localised
        {W}annier functions in metallic systems}, Ann. Henri Poincaré, 20 (2019),
    pp.~1367--1391.
    
    \bibitem[CLPS17]{cances2017robust}
    {\sc E.~Cancès, A.~Levitt, G.~Panati, and G.~Stoltz}, {\em Robust
        determination of maximally localized {W}annier functions}, Phys. Rev. B, 95
    (2017), p.~075114.
    
    \bibitem[DFW18]{drouot2018defect}
    {\sc A.~Drouot, C.~Fefferman, and M.~Weinstein}, {\em Defect modes for
        dislocated periodic media}, arXiv preprint arXiv:1810.05875,  (2018).
    
    \bibitem[Dro18]{drouot2018bulk}
    {\sc A.~Drouot}, {\em The bulk-edge correspondence for continuous dislocated
        systems}, arXiv preprint arXiv:1810.10603,  (2018).
    
    \bibitem[FLTW17]{fefferman2017topologically}
    {\sc C.~Fefferman, J.~Lee-Thorp, and M.~Weinstein}, {\em Topologically
        protected states in one-dimensional systems}, vol.~247, American Mathematical
    Society, 2017.
    
    \bibitem[Hat93]{hatsugai1993chern}
    {\sc Y.~Hatsugai}, {\em Chern number and edge states in the integer quantum
        hall effect}, Phys. Rev. Lett., 71 (1993), p.~3697.
    
    \bibitem[HK11]{hempel2011variational}
    {\sc R.~Hempel and M.~Kohlmann}, {\em A variational approach to dislocation
        problems for periodic {S}chr{\"o}dinger operators}, J. Math. Anal. Appl., 381
    (2011), pp.~166--178.
    
    \bibitem[Kat13]{kato2013perturbation}
    {\sc T.~Kato}, {\em Perturbation theory for linear operators}, vol.~132,
    Springer Science \& Business Media, 2013.
    
    \bibitem[Kor00]{korotyaev2000lattice}
    {\sc E.~Korotyaev}, {\em Lattice dislocations in a 1-dimensional model},
    Commun. Math. Phys., 213 (2000), pp.~471--489.
    
    \bibitem[Kor05]{korotyaev2005schrodinger}
    \leavevmode\vrule height 2pt depth -1.6pt width 23pt, {\em Schr{\"o}dinger
        operator with a junction of two 1-dimensional periodic potentials},
    Asymptotic Analysis, 45 (2005), pp.~73--97.
    
    \bibitem[Pan07]{panati2007triviality}
    {\sc G.~Panati}, {\em Triviality of {B}loch and {B}loch-{D}irac bundles}, Ann.
    Henri Poincaré, 8 (2007), pp.~995--1011.
    
    \bibitem[Pan14]{pankrashkin2014remark}
    {\sc K.~Pankrashkin}, {\em A remark on the discriminant of {H}ill’s equation
        and {H}erglotz functions}, Archiv der Mathematik, 102 (2014), pp.~155--163.
    
    \bibitem[Phi96]{phillips1996self}
    {\sc J.~Phillips}, {\em Self-adjoint {F}redholm operators and spectral flow},
    Can. Math. Bull., 39 (1996), pp.~460--467.
    
    \bibitem[PSB16]{prodan2016bulk}
    {\sc E.~Prodan and H.~Schulz-Baldes}, {\em Bulk and boundary invariants for
        complex topological insulators}, Springer International Publishing, 2016.
    
    \bibitem[PT87]{poschel1987}
    {\sc J.~Pöschel and E.~Trubowitz}, {\em Inverse Spectral Theory}, vol.~130,
    Academic Press, 1987.
    
    \bibitem[RS78]{reed1978analysis}
    {\sc M.~Reed and B.~Simon}, {\em Methods of Modern Mathematical Physics Vol.
        IV: Analysis of Operators}, New York, Academic Press, 1978.
    
    \bibitem[Sim83]{simon1983holonomy}
    {\sc B.~Simon}, {\em Holonomy, the quantum adiabatic theorem, and {B}erry's
        phase}, Phys. Rev. Lett., 51 (1983), p.~2167.
    
\end{thebibliography}

\end{document}